\documentclass[11pt]{article}
\usepackage{multirow}
\usepackage{amsthm,amssymb,amsmath,mathrsfs,cite,float,setspace}
\usepackage{hyperref}
\usepackage{enumerate}
\usepackage[shortlabels]{enumitem}
\usepackage[top=0.8 in, left=1 in, right=1 in, bottom=0.8 in]{geometry}
\usepackage{afterpage, pdflscape}
\usepackage{authblk}

\usepackage[colorinlistoftodos]{todonotes}
\usepackage{makecell}
\usepackage{array}
\newcolumntype{L}[1]{>{\raggedright\let\newline\\\arraybackslash\hspace{0pt}}m{#1}}
\newcolumntype{C}[1]{>{\centering\let\newline\\\arraybackslash\hspace{0pt}}m{#1}}
\newcolumntype{R}[1]{>{\raggedleft\let\newline\\\arraybackslash\hspace{0pt}}m{#1}}

\newcommand{\commentout}[1]{}

\definecolor{blue(pigment)}{rgb}{0.2, 0.2, 0.6}
\newcommand{\hnew}[1]{{\color{black} #1}}
\newcommand{\anew}[1]{{\color{black} #1}}

\usepackage{graphicx,caption,subcaption,epstopdf,enumitem,tikz,standalone}
\usetikzlibrary{shapes,backgrounds}
\usetikzlibrary[positioning,decorations.pathmorphing]
\usetikzlibrary{calc}

\newtheorem {theorem}{Theorem}
\newtheorem {lemma}{Lemma}
\newtheorem {proposition}{Proposition}
\newtheorem {corollary}{Corollary}

\usepackage[noend]{algpseudocode}
\usepackage{algorithm,algorithmicx}

\graphicspath{ {figs/} }

\usepackage{array}
\newcolumntype{P}[1]{>{\centering\arraybackslash}p{#1}}
\newcolumntype{M}[1]{>{\centering\arraybackslash}m{#1}}

\newcommand{\calC}{\mathcal{C}}
\newcommand{\calH}{\mathcal{H}}
\newcommand{\calL}{\mathcal{L}}
\newcommand{\calF}{\mathcal{F}}
\newcommand{\calD}{\mathcal{D}}

\newcommand{\calT}{\mathcal{T}}
\newcommand{\whp}{\alpha} 
\mathchardef\mhyphen="2D
\newcommand{\ecc}{e}

\newcommand{\layer}{\calL}
\newcommand{\rl}[1]{{#1}^{*}}

\title{Injective hulls of various graph classes}
\author[a]{Heather M. Guarnera\footnote{\url{hmichaud@kent.edu}}}
\author[a]{Feodor F. Dragan\footnote{\url{dragan@cs.kent.edu}}}
\author[b]{Arne Leitert\footnote{\url{arne.leitert@cwu.edu}}}
\affil[a]{Department of Computer Science, Kent State University}
\affil[b]{Department of Computer Science, Central Washington University}

\begin{document}

\maketitle
\begin{abstract}
A graph is Helly if its disks satisfy the Helly property, i.e., every family of pairwise intersecting disks in $G$ has a common intersection.
It is known that for every graph $G$, there exists a unique smallest Helly graph $\calH(G)$ into which $G$ isometrically embeds; $\calH(G)$ is called the injective hull of $G$.
Motivated by this, we investigate the structural properties of the injective hulls of various graph classes. We say that a class of graphs $\calC$ is closed under Hellification if  $G \in \calC$ implies $\calH(G) \in \calC$. 
We identify several graph classes that are closed under Hellification. 
We show that permutation graphs are not closed under Hellification, but chordal graphs, square-chordal graphs, and distance-hereditary graphs are.
Graphs that have an efficiently computable injective hull are of particular interest.
A \anew{linear-}time algorithm to construct the injective hull of any distance-hereditary graph is provided and we show that the injective hull of several graphs from some other well-known classes of graphs are impossible to compute in subexponential time. 
In particular, there are split graphs, cocomparability graphs, \hnew{bipartite graphs $G$} such that $\calH(G)$ contains $\Omega(a^{n})$ vertices, where $n=|V(G)|$ and $a>1$.

\medskip

\end{abstract}

\section{Introduction}
In a graph $G=(V(G),E(G))$, a disk $D_G(v,r)$ with radius $r$ and centered at a vertex~$v$ consists of all vertices with distance at most~$r$ from~$v$, i.e., $D_G(v,r) = \{u \in V(G) : d_G(u,v) \le r\}$.
A graph is called \emph{Helly} if every system of pairwise intersecting disks has a non-empty common intersection.
The \emph{injective hull} of an arbitrary graph $G$, denoted $\calH(G)$, is a unique minimal Helly graph which contains $G$ as an isometric subgraph~\cite{isbell,dress,ursLang}.
One measure of how far a graph is from its injective hull is its \emph{Helly-gap}, denoted $\alpha(G)$, which is the minimum integer $k$ such that every vertex of $\calH(G)$ has within distance at most $k$ a vertex of  $G$~\cite{ourManuscriptWeaklyHelly}.
Graphs with a small Helly-gap are precisely the graphs whose disks satisfy the coarse Helly property~\cite{chalopin2020helly,ourManuscriptWeaklyHelly}. 
As it turns out, many well-known graph classes have a small Helly-gap~\cite{ourManuscriptWeaklyHelly,chalopin2020helly} including cube-free median graphs, hereditary modular graphs, 7-systolic complexes, and the graphs of bounded tree-length, bounded hyperbolicity, or bounded $\alpha_i$-metric.

Helly graphs have been well-investigated; they have several characterizations and important features as established in~\cite{FDraganPhD,DraganCenters,Bandelt:1989:DAR:72175.72177,NOWAKOWSKI1983223,QUILLIOT1985186,BANDELT199134,newDDG2020}.  
They are exactly the so-called {\em absolute retracts of reflexive graphs} and
possess a certain elimination scheme~\cite{FDraganPhD,DraganCenters,Bandelt:1989:DAR:72175.72177,NOWAKOWSKI1983223,BANDELT199134} which makes them recognizable in $O(n^2m)$ time~\cite{FDraganPhD}. The Helly property works as a compactness criterion on graphs~\cite{QUILLIOT1985186}.
Many nice properties of Helly graphs are based on the eccentricity $\ecc_G(v)$ of a vertex~$v$, which is defined as the maximum distance from~$v$ to any other vertex of the graph (i.e., $\ecc_G(v) = \max_{u \in V(G)}d_G(v,u)$). The minimum and maximum eccentricity in a graph $G$ is the radius and diameter, respectively.
Conveniently, the eccentricity function in Helly graphs is unimodal~\cite{DraganCenters}, that is, any local minimum coincides with the global minimum.
This fact was recently used in \cite{Dragan2019ASO,DucoffeNEW,newDDG2020} to compute the radius, diameter and a central vertex of a Helly graph in subquadratic time.
Helly graphs can be metrically characterized by the fact that all disks of uniform radius have the Helly property~\cite{newDDG2020}.
Moreover, there are many graph parameters that are strongly related in Helly graphs, including so-called interval thinness, hyperbolicity, pseudoconvexity of disks, and size of the largest isometric subgraph in the form of a square rectilinear grid or a square king grid, among others (cf.~\cite{newDDG2020,DRAGAN2019326}); in particular, a constant bound on any one of these parameters implies a constant bound on all others~\cite{newDDG2020}.


The rich theory behind Helly graphs entices the use of injective hulls as an underlying structure to solve (approximately) problems on~$G$. 
Problems such as finding the diameter or computing vertex eccentricities in $G$ are translatable to finding the diameter of $\calH(G)$ and computing eccentricities of vertices of the Helly graph $\calH(G)$.
Additionally, there is a subquadratic time approximation for radius $rad(G)$ of a graph with an additive error depending on $\alpha(G)$~\cite{DucoffeNEW}.
Moreover, the existence of the injective hull of a graph $G$ is useful to prove properties that appear in $G$.
For example, the existence of injective hulls has been used to prove the existence of a core which intersects  shortest paths for a majority of pairs of vertices, establishing that traffic congestion is inherent in graphs with global negative curvature~\cite{Chepoi:2017:CCI:3039686.3039835}, a.k.a. hyperbolic graphs.
Injective hulls were also used to prove the existence of an eccentricity approximating spanning tree $T$ of $G$ which gives an approximation of all vertex eccentricities with additive error depending 
essentially on $\alpha(G)$~\cite{ourManuscriptWeaklyHelly}. 

\begin{table}[h]
\centering
\begin{tabular}{R{3.5cm}|C{3cm}|C{4cm}}
     Graph Class $\calC$ & \makecell{Closed under \\ Hellification} & \makecell{Hardness to compute \\ $\calH(G)$ for any $G \in \calC$} \\ \hline
     $\delta$-Hyperbolic & Yes & $\Omega(a^n)$ \\
     Chordal             & Yes & $\Omega(a^n)$ \\
     Square-Chordal      & Yes & ? \\
     Distance-Hereditary & Yes & \anew{$O(n + m)$} \\
     Permutation         & No &  ? \\
     Cocomparability     & ? &  $\Omega(a^n)$ \\
     AT-free     & ? &  $\Omega(a^n)$ \\
     \hnew{Bipartite}  &  No  & $\Omega(a^n)$ \\
     (or any triangle-free) &    &  
\end{tabular}\caption{\label{tab}A summary of our results on injective hulls of various graph classes, where $a>1$ and $n=|V(G)|$. "?" means that this question is still open.}
\end{table} 

The importance of $\calH(G)$ as an underlying structure drives our interest in the injective hulls of various graph classes.
Our main contributions are summarized in Table~\ref{tab} and organized as follows.
We identify in Section~\ref{sec:injectiveHullProperties} several universal properties of the injective hull of any graph.
Next, we focus on a graph $G$ that belongs to a particular graph class $\calC$.
In particular, we are interested in whether $\calC$ is closed under Hellification, i.e., whether $G \in \calC$ implies $\calH(G) \in \calC$.
In Section~\ref{sec:hyperbolicGraphs}, we give a graph theoretic proof that hyperbolic graphs are closed under Hellification.
Moreover, we prove that satisfying the Helly property in disks of radii at most $\delta+1$ is sufficient to satisfy the Helly property in all disks of a $\delta$-hyperbolic graph.
In Section~\ref{sec:ATFree}, we show that permutation graphs are not closed under Hellification and provide conditions in which AT-free graphs are.
In Section~\ref{sec:chordal}, we prove that chordal graphs and square-chordal graphs are closed under Hellification.
In Section~\ref{sec:dhg}, we add distance-hereditary graphs to the growing list of graph classes closed under Hellification and provide a \anew{linear-}time algorithm to compute $\calH(G)$ of a distance-hereditary graph $G$.
We demonstrate in Section~\ref{sec:exponentialInjectiveHull} that the injective hull of several graphs from some other well-known classes of graphs are impossible to compute efficiently. 
Specifically, there is a graph $G$ where the number of vertices in $\calH(G)$ is $\Omega(a^n)$, where $a > 1$, $n=|V(G)|$.
\hnew{We construct three such graphs: a split graph, a cocomparability graph, and a bipartite graph. } Note also that such well-known graph classes as interval graphs, strongly chordal graphs, dually chordal graphs are subclasses of Helly graphs~\cite{FDraganPhD,DBLP:journals/siamdm/BrandstadtDCV98,draganLocationProblems} and therefore for them trivially $\calH(G)$ coincides with $G$. The definitions of graph classes not provided here can be found in \cite{doi:10.1137/1.9780898719796}.

\section{Preliminaries}\label{sec:preliminaries}
All graphs $G=(V(G),E(G))$ occurring in this paper are undirected, connected, and without loops or multiple edges.
A \emph{path} $P(v_0,v_k)$ is a sequence of vertices $v_0,\dots,v_k$ such that $v_iv_{i+1} \in E$ for all $i \in [0,k-1]$; its \emph{length} is $k$.
The \emph{distance} $d_G(u,v)$ between two vertices~$u$ and~$v$ is the length of a shortest path connecting them in $G$; the distance $d_G(u,S)$ between a vertex $u$ and a set of vertices $S \subseteq V(G)$ is the minimum distance from~$u$ to any vertex of~$S$.
The interval $I(x,y)$ between vertices $x,y$ is the set of all vertices belonging to a shortest $(x,y)$-path, i.e., $I(x,y) = \{v \in V(G) : d_G(x,y) = d_G(x,v) + d_G(v,y)\}$.
The interval slice $S_k(x,y)$ is the set of vertices belonging to $I(x,y)$ and at distance $k$ from~$x$, i.e., $S_k(x,y)= \{v \in I(x,y) : d_G(x,v) = k \}$.
The {neighborhood of $v$} consists of all vertices adjacent to $v$, denoted by $N(v)$.
The \emph{closed neighborhood} of $v$ is defined as $N[v]=N(v) \cup \{v\}$.
The degree $deg(v)$ of a vertex $v$ is the number of neighbors it has, i.e., $deg(v)=|N(v)|$.
A vertex $v$ is \emph{pendant} if $deg(v)=1$.
Two vertices~$v$ and~$u$ are \emph{twins} if they have the same neighborhood.
\emph{True twins} are adjacent; \emph{false twins} are not.
A disk $D_G(v,r)$ with radius $r$ and centered at a vertex~$v$ consists of all vertices with distance at most~$r$ from~$v$, i.e., $D_G(v,r) = \{u \in V(G) : d_G(u,v) \le r\}$.
A set $M \subset V(G)$ is said to \emph{separate} a vertex pair $x,y \in V(G)$ if the removal of $M$ from $G$ separates $x$ and $y$ into distinct connected components.
The \emph{eccentricity} of a vertex~$v$ is defined as $\ecc_G(v) = \max_{u \in V(G)}{d_G(v,u)}$.
The radius $rad(G)$ and diameter $diam(G)$ are the minimum and maximum eccentricity, respectively. 
The $k^{th}$ \emph{power} $G^k$ of a graph $G$ is a graph that has the same set of vertices, but in which two distinct vertices are adjacent if and only if their distance in $G$ is at most $k$.
A subgraph $G'$ of a graph $G$ is called {\em isometric} if for any two vertices $x,y$ of $G'$, $d_G(x,y)=d_{G'}(x,y)$ holds.
We denote by $\langle S \rangle$ the subgraph of~$G$ induced by the vertices~$S \subset V$.
The subindex $G$ is omitted when the graph is known by context.

A \emph{chord} of a path (cycle) $v_0,\dots,v_k$ is an edge between two vertices of the path (cycle) that is not an edge of the path (cycle).
A set $M \subseteq V(G)$ is an \emph{independent set} if for all $u,v \in V(G)$, $uv \notin E(G)$.
A set $M \subseteq V(G)$ is a \emph{clique} (or \emph{complete subgraph}) if all distinct vertices $u,v \in M$ have $uv \in E(G)$. A set $M \subseteq V(G)$ is said to be a \emph{2-set} if for every $x,y\in M$, $d(x,y)\le 2$ holds. A 2-set $M$ is \emph{maximal} in $G$ if it is maximal by inclusion. 
A vertex $v$ is said to \emph{suspend} a set $M \subseteq V(G)$ if $vu \in E(G)$ for each $u \in M\setminus \{v\}$; $v$ is also said to be \emph{universal} to $M\setminus \{v\}$.
We denote by $C_k$ a cycle induced by $k$ vertices, by $W_k$ an induced wheel of size $k$, i.e., a cycle $C_k$ with one additional vertex universal to $C_k$, and by $K_n$ a clique of $n$ vertices.
A graph $B$ is \emph{bipartite} if its vertex set can be partitioned into two independent sets $X$ and $Y$, i.e., each edge $uv \in E(B)$ has one end in $X$ and the other in $Y$.

A \emph{tree-decomposition} $(\calT,T)$ for a graph $G$ is a family $\calT = \{ B_1, B_2, \dots \}$ of subsets of $V(G)$, called \emph{bags}, such that $\calT$ forms a tree $T$ with the bags in $\calT$ as nodes which satisfy the following conditions:
(i) each vertex is contained in a bag,
(ii) for each edge $uv \in E(G)$, $\calT$ has a bag $B$ with $u,v \in B$, and
(iii) for each vertex $v \in V(G)$, the bags containing $v$ induce a subtree of $T$.
A tree decomposition has \hnew{ \emph{breadth}} $\rho$ if, for each bag $B$, there is a vertex $v$ in $G$ such that $B \subseteq D_G(v,\rho)$.
A tree decomposition has \hnew{ \emph{length}} $\lambda$ if the diameter in~$G$ of each bag $B$ is at most $\lambda$.
The \emph{tree-breadth} $tb(G)$~\cite{DBLP:journals/algorithmica/DraganK14} and  \emph{tree-length} $tl(G)$~\cite{DOURISBOURE20072008} are the minimum breadth and length, respectively, among all possible tree decompositions of $G$.

A graph $G$ is \emph{Helly} if, for any system of disks $\calF = \{ D(v, r(v)) : v \in S \subseteq V(G)\}$, the following Helly property holds:
if $X \cap Y \neq \emptyset$ for every $X,Y \in \calF$,
then $\bigcap_{v \in S} D(v, r(v)) \neq \emptyset$.
\emph{Pseudo-modular} graphs are a far-reaching superclass of Helly graphs. By definition, a graph $G$ is \emph{pseudo-modular}  if every triple $x,y,z$ of its vertices admits either a `median' vertex or a `median' triangle, i.e., either there is a vertex $v$ such that $d(x,y)=d(x,v)+d(v,y), d(x,z)=d(x,v)+d(v,z), d(z,y)=d(z,v)+d(v,y)$ or there is a triangle (three pairwise adjacent vertices) $v,u,w$ such that $d(x,y)=d(x,v)+1+d(u,y), d(x,z)=d(x,v)+1+d(w,z), d(z,y)=d(z,w)+1+d(u,y)$. 
Pseudo-modular graphs are characterized as follows.

\begin{proposition} \cite{Bandelt:1986:PG:10100.10102} \label{prop:pseudo-modular-def}
For a connected graph $G$ the following are equivalent:
  \begin{enumerate}[topsep=0pt,itemsep=20pt,parsep=0pt,partopsep=0pt]
  	\setlength\itemsep{0em}
    \item[i)] $G$ is pseudo-modular.
    \item[ii)] Any three pairwise intersecting disks of $G$ have a nonempty intersection.
    \item[iii)] If $1 \leq d(v,w) \leq 2$ and $d(u,v) = d(u,w) = k \geq 2$ for vertices $u,v,w$ of $G$,
    then there exists a vertex $x$ such that $d(v,x) = d(w,x) = 1$ and $d(u,x)=k-1$.
  \end{enumerate}
\end{proposition}

The presence of pseudo-modularity in a graph $G$ is of algorithmic interest because it limits the number of disk families which must satisfy the Helly property for $G$ to be considered Helly.
Specifically, a pseudo modular graph is Helly if and only if it is \emph{neighborhood-Helly}, i.e., if the family of its all unit disks (all closed neighborhoods) $\{D(v,1) : v \in V(G)\}$   satisfies the Helly property. 

\begin{proposition} \cite{Bandelt:1989:DAR:72175.72177} \label{prop:pseudo-modular-helly}
$G$ is Helly if and only if it is pseudo-modular and neighborhood-Helly.
\end{proposition}

It is clear that $G$ is neighborhood-Helly if and only if all maximal 2-sets of $G$ are suspended. 




We define the remaining graph classes in their corresponding sections; the definitions of graph classes not provided here can be found in \cite{doi:10.1137/1.9780898719796}.

\section{Injective hulls}\label{sec:injectiveHullProperties}
By an equivalent definition of an injective hull \cite{dress} (also called a tight span),
each vertex $f \in V(\calH(G))$ can be represented as a vector with values $f(x)$ for each $x \in V(G)$,
such that the following two properties hold:
\begin{equation} \label{eq:atLeastDistance}
\forall x,y \in V(G) \ f(x) + f(y) \geq d(x,y) 
\end{equation}
\vspace*{-3mm}
\begin{equation} \label{eq:extremalFunctions}
\forall x \in V(G) \ \exists y \in V(G) \ f(x) + f(y) = d(x,y) 
\end{equation}

Additionally, there is an edge between two vertices $f,g \in V(\calH(G))$ if and only if their Chebyshev distance is 1, i.e., $\max_{x \in V(G)} \lvert f(x) -g(x) \rvert = 1$.
Thus, $d_{\calH(G)}(f,g) = \hnew{\max_{x \in V(G)}} \lvert f(x) - g(x) \rvert$.
Notice that if $f \in V(\calH(G))$, then $\{D(x, f(x)) : x \in V(G)\}$ is a family of pairwise intersecting disks.
For a vertex~$z \in V(G)$, define the distance function $d_z$ by setting $d_z(x)=d_G(z,x)$ for any $x \in V(G)$.
By the triangle inequality, each $d_z$ belongs to $V(\calH(G))$.
An isometric embedding of~$G$ into~$\calH(G)$ is obtained by mapping each vertex~$z$ of~$G$ to its distance vector $d_z$.

We classify every vertex $v$ in $V(\calH(G))$ as either a real vertex or a Helly vertex.
A vertex~$f \in V(\calH(G))$ is a \emph{real vertex} provided $f=d_z$ for some $z \in V(G)$, i.e., there is a one-to-one correspondence between $z \in V(G)$ and its representative real vertex $f \in V(\calH(G))$ which uniquely satisfies $f(z) = 0$ and $f(x) = d_G(z,x)$ for all $x \in V(G)$.
By an abuse of notation, we will interchangeably use $V(G)$ to represent the vertex set in~$G$ as well as the vertex subset of~$\calH(G)$ which uniquely corresponds to the vertex set of~$G$.
Then, a vertex~$v \in V(\calH(G))$ is a real vertex if it belongs to $V(G)$ and a \emph{Helly vertex} otherwise.
Equivalently, a vertex~$h \in V(\calH(G))$ is a Helly vertex provided that $h(x) \ge 1$ for all $x \in V(G)$, that is, a Helly vertex exists only in the injective hull~$\calH(G)$ and not in~$G$.
A path $P(x,y)$ in $\calH(G)$ connecting vertices $x,y \in V(G)$ is said to be a \emph{real path} if each vertex $u \in P(x,y)$ is real.
We often use the terms \emph{Hellify} (verb) and \emph{Hellification} (noun) to describe the process by which edges and Helly vertices are added to $G$ to construct $\calH(G)$.
When $G$ is known by context, we often let $H := \calH(G)$.

A vertex $x$ is a \emph{peripheral} vertex if $I(y,x) \not\subset I(y,z)$ for some vertex $y$ and all vertices $z \neq x$. In $\calH(G)$, all peripheral vertices are real. Consequently, all farthest vertices from any $v \in V(\calH(G))$ are real. It follows that, in $\calH(G)$, 
any shortest path is a subpath of a shortest path between real vertices.

\begin{proposition}~\cite{ourManuscriptWeaklyHelly} \label{prop:outsidePointsAreReal}
  Peripheral vertices of $\calH(G)$ are real.
\end{proposition}



\hnew{
The following result was proven earlier in~\cite{ourManuscriptWeaklyHelly} only for $H := \calH(G)$.
For completeness, we provide a proof that it holds for any host $H$ such that $G$ embeds isometrically into $H$ and all peripheral vertices in $H$ are from $G$.

\begin{proposition}~\cite{ourManuscriptWeaklyHelly} \label{prop:shortestPathSubsetOfReal}
Let $H$ be a host such that $G$ embeds isometrically into $H$ and all peripheral vertices in $H$ are from $G$.
For any shortest path $P(x,y)$, where $x,y \in V(H)$, there is a shortest path $P(\rl{x},\rl{y})$, where $\rl{x},\rl{y} \in V(G)$ are peripheral vertices of $G$, such that $P(\rl{x},\rl{y}) \supseteq P(x,y)$.
\end{proposition}
\begin{proof}
If $x$ and $y$ are both real vertices, then the proposition is trivially true.
Without loss of generality, suppose vertex $y$ does not belong to $V(G)$.
Consider a breadth-first search layering where $y$ belongs to layer $L_i$ of BFS($H, x$).
Let $y' \in L_k$ be a vertex with $y \in I(x,y')$ that maximizes $k=d_H(x,y')$.
Then, for any vertex $z \in V(H)$, $I(x,y') \not\subset I(x,z)$.
Hence, $y'$ is a peripheral vertex; by assumption, $y' \in V(G)$.
If $x \notin V(G)$, then applying the previous step using BFS($H, y'$) yields vertex $x' \in V(G)$.
\end{proof}
}

Let the distance $d(z,P)$ from a vertex~$z$ to an $(x,y)$-path~$P$ be the minimum distance from~$z$ to any vertex $u \in P$.
We next show that for any vertex $z \in V(\calH(G))$ and any $(x,y)$-path $P$ in $\calH(G)$, there is a real $(\rl{x},\rl{y})$-path $\rl{P}$ in $G$ which behaves similarly to $P$ with respect to some distance properties.
In particular, we show that if $x,y,z \in V(G)$ then for every $(x,y)$-path $P$ in $\calH(G)$, there is a real $(x,y)$-path $\rl{P}$ in $G$ such that $d(z,\rl{P}) \ge d(z,P)$.

We have the following lemma.

\begin{lemma}\label{lem:edgeToAvoidingPathInG}
\hnew{Let $H$ be the injective hull of $G$.}
For any vertex $z \in V(H)$ and edge $xy \in E(H)$, there is a real $(x^*, y^*)$-path $P^*$ in $G$
such that $d_H(z,P^*) \geq d_H(z, \{x,y\})$, $I(z,x) \subseteq I(z,x^*)$, and $I(z,y) \subseteq I(z,y^*)$.
\end{lemma}
\begin{proof}
Let $\layer_0, \layer_1, \dots, \layer_{\ecc(z)}$ be layers of $H$ produced by a breadth-first search rooted at vertex $z$.
Without loss of generality, let $d_H(z,\{x,y\}) = d_H(z,x) = k$.
Hence, $x \in \layer_k$ and $y \in \layer_p$ where $p=k$ or $p=k+1$.
By Proposition~\ref{prop:shortestPathSubsetOfReal},
there is a vertex $x^* \in V(G)$ such that $x \in I(z,x^*)$ and
there is a vertex $y^* \in V(G)$ such that $y \in I(z,y^*)$.
Then, $x^* \in \layer_j$ for some $j \ge k$ and $y^* \in \layer_\ell$ for some $\ell \geq p$.
Since $G$ is isometric in $H$, there is a shortest $(x^*,y^*)$-path $P^*$ in $G$ of length $d_H(x^*,y^*)$ consisting of all real vertices, as illustrated in Figure~\ref{fig:avoidingPathInG}(a).
By the triangle inequality, $d_G(x^*, y^*) = d_H(x^*, y^*) \leq d_H(x^*, x) + 1 + d_H(y,y^*) \leq (j - k) + 1 + (\ell - p)$.
By contradiction, assume there is a vertex $w^* \in P^*$ such that $d_H(z,w^*) < k$.
Then, $d_H(x^*, y^*) = d_H(x^*, w^*) + d_H(w^*, y^*) \geq (j - k + 1) + (\ell - k + 1)$, a contradiction.
\end{proof}

\begin{figure}
\centering
\begin{subfigure}{.5\textwidth}
  \centering
  \includegraphics[scale=0.6]{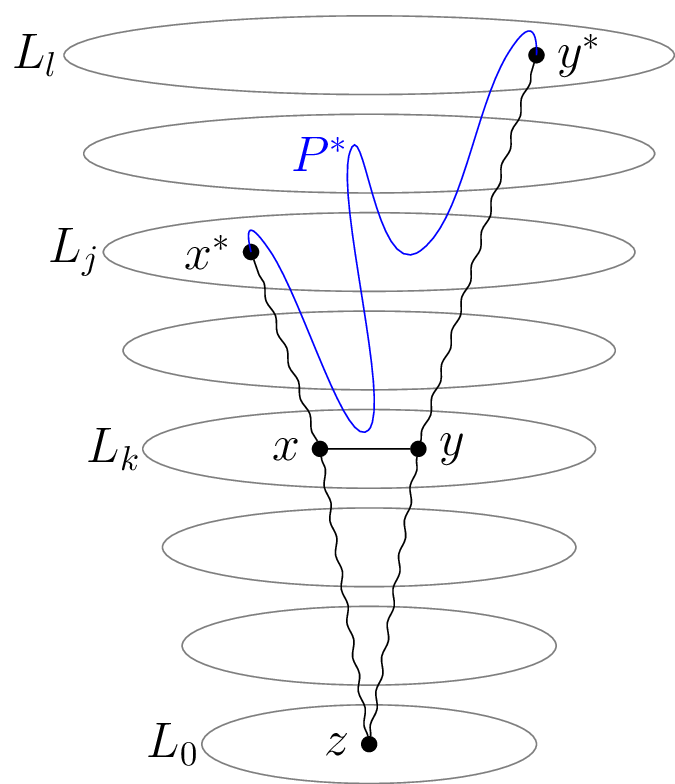}
  \centering\caption{}
  \label{fig:edgeToAvoidingPathInG}
\end{subfigure}%
\begin{subfigure}{.5\textwidth}
  \centering
  \includegraphics[scale=0.6]{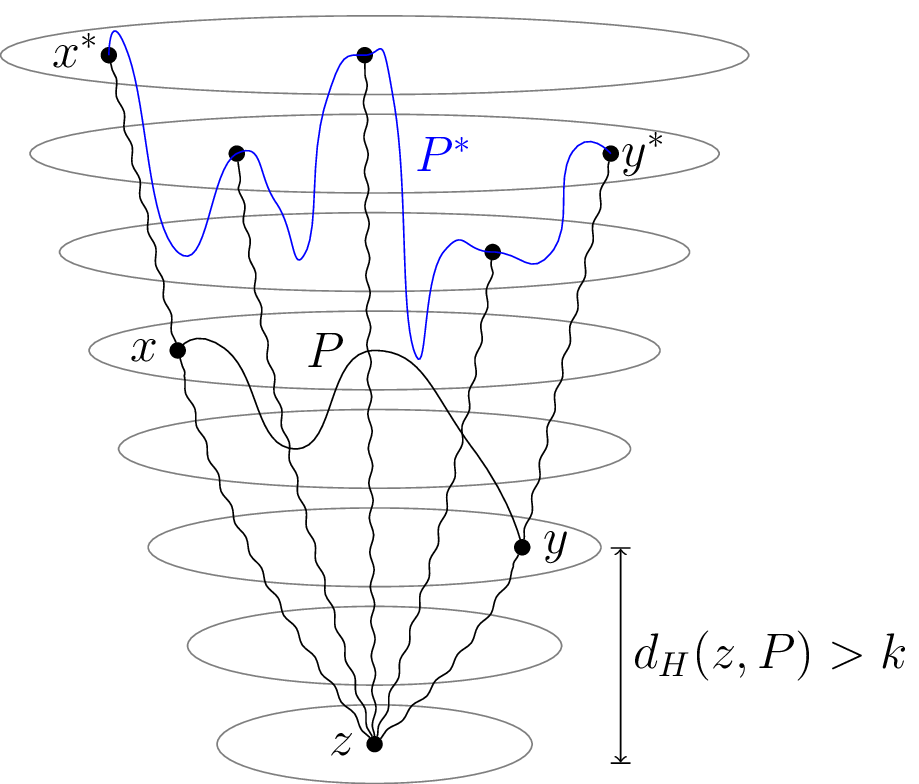}
  \centering\caption{}
  \label{fig:pathToAvoidingPathInG}
\end{subfigure}
\caption{Illustration to the proofs of (a) Lemma~\ref{lem:edgeToAvoidingPathInG} and (b) Theorem~\ref{thm:diskSeparatesInH}, where real paths are shown in blue.}
\label{fig:avoidingPathInG}
\end{figure}

\begin{theorem} \label{thm:diskSeparatesInH}
\hnew{Let $H$ be the injective hull of $G$.}
For any $x,y,z \in V(G)$,
the disk $D_G(z,k)$ separates vertices $x,y$ in $G$ if and only if disk $D_H(z,k)$ separates vertices $x,y$ in $H$.
\end{theorem}
\begin{proof}
($\leftarrow$) It suffices to remark that if $D_G(z,k)$ does not separate $x,y$ in $G$ due to a path $\rl{P}$ connecting them, then the same path establishes that $D_H(z,k)$ does not separate $x,y$ in $\calH(G)$.

($\rightarrow$) Suppose the disk~$D_H(z,k)$ does not separate vertices~$x,y$ in~$H$ and assume, without loss of generality, that $D_H(z,k)\cap \{x,y\}=\emptyset$. Then, there is an $(x,y)$-path~$P$ in~$H$ such that $d_H(z,P) > k$.
Let $P=v_0,v_1,v_2,\dots,v_j$, where $v_0 := x$ and $v_j := y$.
By Lemma~\ref{lem:edgeToAvoidingPathInG}, for each edge $v_iv_{i+1}$ on~$P$,
there is a real $(\rl{v_i}, \rl{v_{i+1}})$-path $\rl{P_i}$ in~$G$ such that
$d_H(z, \rl{P_i}) \geq d_H(z, \{v_i,v_{i+1}\}) > k$,
as shown in Figure~\ref{fig:avoidingPathInG}(b).
Let $\rl{P}$ be the real path obtained by joining, for $i \in [0, j-1]$, each real path $P^*_i$ by their end vertices.
Then, $d_H(z, \rl{P}) \geq d_H(z, P) > k$.
As a result, the disk $D_G(z,k)$ does not separate vertices~$x,y$ in~$G$.
\end{proof}

\begin{corollary}\label{cor:diskSeparatesInH}
\hnew{Let $H$ be the injective hull of $G$.}
For any $x,y,z \in V(G)$ and every $(x,y)$-path $P$ in $H$, there is a real $(x,y)$-path $\rl{P}$ in $G$ such that $d(z,\rl{P}) \ge d(z,P)$.
\end{corollary}

\section{$\delta$-Hyperbolic graphs}\label{sec:hyperbolicGraphs}
A metric space $(X,d)$ is $\delta$-hyperbolic if it satisfies Gromov's 4-point condition: for any four points $u,v,w,x$ from $X$ the two larger of the three distance sums $d(u,v) + d(w,x)$, $d(u,x) + d(v,w)$, and $d(u,w) + d(v,x)$ differ by at most $2\delta \geq 0$.
A connected graph equipped with the standard graph metric $d_G$ is $\delta$-hyperbolic if the metric space $(V,d_G)$ is $\delta$-hyperbolic. The smallest value~$\delta$ for which~$G$ is $\delta$-hyperbolic is called the \emph{hyperbolicity} of~$G$ and is denoted $\delta(G)$. Note that $\delta(G)$ is an integer or a half-integer. 
For a quadruple of vertices $u,v,w,x \in V(G)$, it will be convenient to denote by $hb(u,v,w,x)$ half the difference of the largest two distance sums among $d(u,v) + d(w,x)$, $d(u,x) + d(v,w)$, and $d(u,w) + d(v,x)$.

It is known~\cite{isbell,ursLang} that the hyperbolicity of any metric space is preserved in its injective hull. For completeness, we provide a graph-theoretic proof of this result and show that in fact hyperbolicity is preserved in any host $H$ as long as distances in $G$ are preserved in $H$ and that peripheral vertices of $H$ are real.

\begin{proposition} \label{prop:hyperbolicityEqualGeneric}
If $H$ is a host graph such that $G$ embeds isometrically into $H$ and all peripheral vertices in $H$ are from $G$,
then $\delta(G) = \delta(H)$.
\end{proposition}
\begin{proof}
As $G$ embeds isometrically into $H$, $\delta(G) \le \delta(H)$.
By contradiction, assume $\delta(H) > \delta(G)$.
Let $x,y,z,t \in V(H)$ with $hb(x,y,z,t) > \delta(G)$ such that $|V(G) \cap \{x,y,z,t\}|$ is maximized.
Without loss of generality, let $d_H(x,t) + d_H(z,y) \geq d_H(x,z) + d_H(t,y) \geq d_H(x,y) + d_H(z,t)$. 
If $\{x,y,z,t\} \subseteq V(G)$, then $hb(x,y,z,t) \leq \delta(G)$, a contradiction.
Thus, without loss of generality, suppose $x \notin V(G)$.
By Proposition~\ref{prop:shortestPathSubsetOfReal}, there is a peripheral vertex $\rl{x} \in V(G)$ such that $I(t,x) \subset I(t,\rl{x})$ for vertex $t \in V(H)$.
Let $d_H(t,\rl{x}) = d_H(t,x) + \gamma$.
Clearly, $d_H(\rl{x},t) + d_H(z,y) \ge \max\{d_H(\rl{x},y) + d_H(z,t), d_H(\rl{x},z) + d_H(t,y)\}$.

Suppose that $d_H(\rl{x},y) + d_H(z,t) \geq d_H(\rl{x},z) + d_H(t,y)$.
By the triangle inequality and definition of hyperbolicity, we have
\begin{align*}
2hb(\rl{x},y,z,t) &= d_H(\rl{x},t) + d_H(z,y) - d_H(\rl{x},y) - d_H(z,t) \\
&\ge d_H(x,t) + d_H(z,y) + \gamma - d_H(x,y) - d_H(z,t) - \gamma \\
&= d_H(x,t) + d_H(z,y) - d_H(x,y) - d_H(z,t) \\
&\ge d_H(x,t) + d_H(z,y) - d_H(x,z) - d_H(t,y) \\
&= 2hb(x,y,z,t).
\end{align*}
Thus, $hb(\rl{x},y,z,t) \geq hb(x,y,z,t)$, a contradiction with the maximality of the number of real vertices in the quadruple.

Suppose now that $d_H(\rl{x},z) + d_H(t,y) \geq d_H(\rl{x},y) + d_H(z,t)$.
By the triangle inequality and definition of hyperbolicity, we have
\begin{align*}
2hb(\rl{x},y,z,t) &= d_H(\rl{x},t) + d_H(z,y) - d_H(\rl{x},z) - d_H(t,y) \\
&\geq d_H(x,t) + d_H(z,y) + \gamma - d_H(x,z) - d_H(t,y) - \gamma \\
&= d_H(x,t) + d_H(z,y) - d_H(x,z) - d_H(t,y) \\
&= 2hb(x,y,z,t).
\end{align*}
Thus, $hb(\rl{x},y,z,t) \geq hb(x,y,z,t)$, again a contradiction with the maximality of the number of real vertices in the quadruple.
\end{proof}

\begin{theorem} \label{cor:hyperbolicityEqual}
For any graph $G$, $\delta(G) = \delta(\calH(G))$. That is, the $\delta$-hyperbolic graphs are closed under Hellification. 
\end{theorem}

We next show that a $\delta$-hyperbolic graph $G$ is Helly if its disks up to radii $\delta+1$ satisfy the Helly property.
In this sense, a localized Helly property implies a global Helly property, akin to what is known for pseudo-modular graphs wherein  all disks of radii at most 1 satisfy the  Helly property implies all disks (of all radii) satisfy the  Helly property.


\begin{lemma} \label{lem:hyperbolicity-disk-locality}
  If $G$ is $\delta$-hyperbolic and all disks with up to $\delta+1$ radii
  satisfy the Helly property, then $G$ is a Helly graph.
\end{lemma}
\begin{proof}
  Assume all disks with radii at most $\delta+1$ satisfy the Helly property.
  Clearly $G$ is neighborhood-Helly.
  By Proposition~\ref{prop:pseudo-modular-helly}, it remains only to prove that $G$ is pseudo-modular.
  We apply Proposition \ref{prop:pseudo-modular-def}(iii).
  Consider three vertices $u,v,w$ such that $d(u,v)=d(u,w)=k\geq2$, and either $v$ and $w$ are adjacent or have a common neighbor $z$.
  We claim that $d(u,v) = d(u,w) = k$ implies there is a vertex $t$ adjacent to $v$ and $w$ and at distance $k-1$ from $u$.
  We use an induction on $d(u,v)$.
  By assumption, it is true for $k \leq \delta+2$ as the pairwise-intersecting disks $D(u,k-1)$, $D(v,1)$, $D(w,1)$ have a common vertex $t$ by the Helly property.

  Consider the case when $d(u,v) = d(u,w) = k > \delta+2$.
  Let $x \in I(v,u)$ and $y \in I(w,u)$ be vertices such that $d(x,u)=d(y,u) = \delta+2$.
  We claim the disks $D(x, \delta+1)$, $D(y, \delta+1)$, and $D(u,1)$ pairwise intersect; then, vertex $u^*$ exists by the Helly property and applying the inductive hypothesis to vertex $u^*$ equidistant to $v,w$ yields the desired vertex $t$.
  Clearly, $D(u,1)$ intersects both $D(x,\delta+1)$ and $D(y,\delta+1)$.
  It remains to show that $d(x,y) \leq 2\delta + 2$.
  \commentout{
  {\em Case 1. Suppose $vw \notin E(G)$.}
  As $d(v,w) \le 2$, there is a neighbor $z$ adjacent to $v,w$.
  By the triangle inequality and $d(u,v)$ distance requirements, $k-1 \le d(u,z) \le k+1$.
  If $d(u,z)=k-1$, we are done as vertex $t := z$ satisfies the desired properties.
  Assume that $d(u,z) \ge k$.
  If $d(x,z)=d(y,z)$, then $x,y \in S_{\delta+1}(u,z)$; by Lemma~\ref{half-thin}, $d(x,y) \leq 2\delta$ and we are done.
  Without loss of generality, let $d(x,z)=k$ and $d(y,z)=k+1$.
  We apply the 4-point condition to vertices $u,x,w,y$.
  Observe that $d(u,x)+d(w,y)=k$, $d(u,w)+d(x,y)=k+d(x,y)$, and $d(u,y)+d(w,x)=k+1$.
  Clearly, the first sum is the smallest.
  If the last sum is largest, then $d(x,y) \le 1$ and we are done.
  If the middle sum is largest, then by the 4-point condition $2\delta \ge d(u,w)+d(x,y)-d(u,y)-d(w,x)=d(x,y)+1$;
  hence, $d(x,y) \le 2\delta+1$.
  
  {\em Case 2. Suppose $vw \in E(G)$.}
  As in the previous case, we apply the 4-point condition to vertices $u,x,w,y$. 
  Observe that $d(u,x)+d(w,y)=k$, $d(u,w)+d(x,y)=k+d(x,y)$, and $d(u,y)+d(w,x)=k+1$.
  Hence, $d(x,y) \le 2\delta+1$.
  }

 Consider vertices $u,x,y,w$ and three distance sums: 
 $A:=d(u,w)+d(x,y)$, $B:=d(u,y)+d(x,w)$ and $C:=d(u,x)+d(y,w)$.  We have $A=k+d(x,y)$ and $C=k$.
 Moreover, $k\le B\le k+2$ as $k=d(u,w)\le   d(u,x)+d(x,v)+d(v,w)\le d(u,v)+2=k+2$.  Hence, $C$ is a smallest sum. If $B\ge A$ then $k+2\ge B\ge A=k+d(x,y)$ implies $d(x,y)\le 2\le 2\delta +2$.  
 If $A\ge B$ then, by 4-point condition, $2\delta\ge A-B\ge k+d(x,y)-k-2$, i.e., $d(x,y)\le 2\delta+2$. 
\end{proof}

\section{Permutation graphs and relatives}\label{sec:ATFree}
Permutation graphs can be defined as follows. Consider two parallel lines (upper and lower) in the plane. Assume that each line contains $n$ points, labeled 1 to $n$, and each two points with the same label define a segment with that label. The intersection graph of such a set of segments between two parallel lines is called a \emph{permutation graph}~\cite{doi:10.1137/1.9780898719796}.
An \emph{asteroidal triple} is an independent set of three vertices such that each pair is joined by a path that avoids the closed neighborhood of the third.
A far reaching superclass of permutation graphs are the \hnew{\emph{AT-free} graphs}, i.e., the graphs that do not contain any asteroidal triples~\cite{corneilATFree}.

We show that permutation graphs are not closed under Hellification.
Moreover, if the Helly-gap of some AT-free graph is 2, then AT-free graphs are also not closed under Hellification.

\begin{figure}[h]
\begin{minipage}{\textwidth}
\begin{center}
\includegraphics[scale=0.7]{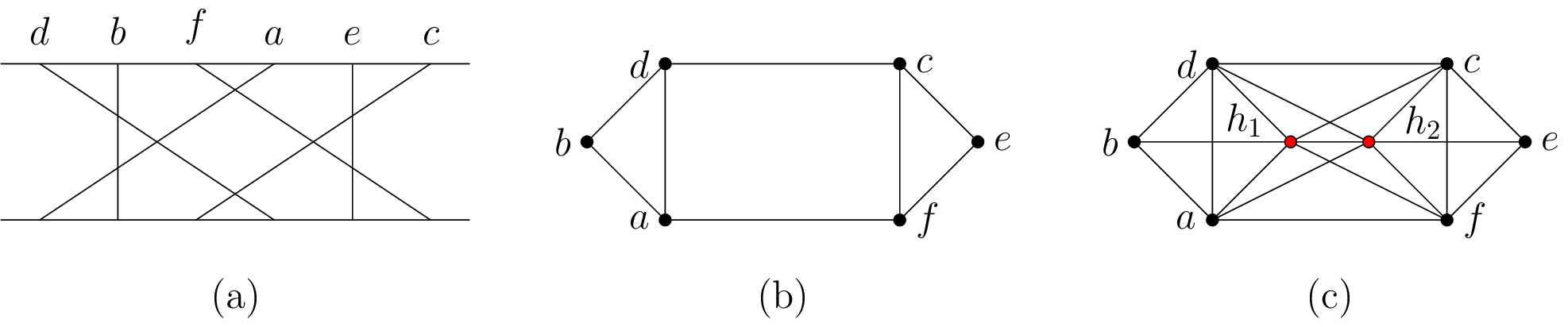}
\captionof{figure}{A permutation model \hnew{(a)} corresponding to permutation graph $G$ \hnew{(b)} and its injective hull $\calH(G)$ \hnew{(c)}, where $\calH(G)$ is not a permutation graph.}\label{fig:permutationNotClosed}
\end{center}
\end{minipage}
\end{figure}

\begin{lemma}
Permutation graphs are not closed under Hellification.
\end{lemma}
\begin{proof}
The graph $G$ illustrated in Figure~\ref{fig:permutationNotClosed} is an example of a permutation graph $G$ for which $\calH(G)$ is not a permutation graph (although, $\calH(G)$ is AT-free).
Note that only two Helly vertices $h_1$ and $h_2$ are added to produce $\calH(G)$, where $h_1$ is adjacent to real vertices $b,d,a,c,f$ and $h_1e \notin E(\calH(G))$. The resulting graph $\calH(G)$ is not a permutation graph since such a vertex/segment $h_1$ cannot be added to the essentially unique permutation model of $G$ depicted in Figure~\ref{fig:permutationNotClosed}; any segment $h_1$ intersecting the segments $b,d,a,c,f$ needs to intersect also the segment $e$.
\end{proof}

For an AT-free graph $G$, the Helly-gap $\alpha(G)$ is impacted by whether $\calH(G)$ is AT-free.
Recall that the Helly gap $\alpha(G)$ is the minimum integer $\alpha$ such that the distance from any Helly vertex $h \in V(H)$ to a closest real vertex $x \in V(G)$ is at most $\whp$.
It is known~\cite{ourManuscriptWeaklyHelly} that any AT-free graph $G$ has $\whp(G) \le 2$.

\begin{lemma} \label{lem:ATFreeAndHDef}
For any graph $G$, $\alpha(G) \le 1$ if $\calH(G)$ is AT-free.
\end{lemma}
\begin{proof}
By contradiction, suppose $\alpha(G) \ge 2$ for some graph $G$ and $H := \calH(G)$ is AT-free.
Then, there is a vertex $h \in V(H)$ such that $d_{H}(h,v) \ge \alpha(G)$ for all $v \in V(G)$.
Let $x \in V(G)$ be closest to $h$; then, $d_{H}(h,x)=\alpha(G) \ge 2$.
By Proposition~\ref{prop:shortestPathSubsetOfReal}, there is a real vertex $y \in V(G)$ such that $h \in I(x,y)$. Moreover, $d_H(h,y) \geq d_H(h,x) \ge 2$.
Let $P$ be a shortest $(x,y)$-path of $H$ with $h \in P$.
As $G$ is isometric in $H$,
there is a (real) shortest $(x,y)$-path $P^*$ in $G$.
By $d_H(x,y)$ distance requirements, \hnew{all shortest $(h,y)$-paths avoid $N[x]$, and all shortest $(h,x)$-paths avoid $N[y]$.
As $P^* \subseteq V(G)$ and $\alpha(G) \ge 2$, then $P^*$ also avoids $N[h]$.}
Therefore, $\{x,y,h\}$ forms an asteroidal triple in $H$, a contradiction.
\end{proof}

\begin{corollary}
If there is an AT-free graph $G$ with $\alpha(G)=2$, then AT-free graphs are not closed under Hellification.
\end{corollary}

Currently, we do not know whether there is an AT-free graph $G$ with $\alpha(G)=2$. 

\section{Chordal Graphs and Square-Chordal Graphs}\label{sec:chordal}
A graph is \emph{chordal} if it contains no induced cycle $C_k$ of length $k \ge 4$. A graph $G$ is \emph{square-chordal} if $G^2$ is chordal.
%
In this section, we will show that for a chordal (square-chordal) graph $G$, 
its injective hull $\calH(G)$ is also chordal (square-chordal).
That is, chordal graphs and square-chordal graphs are closed under Hellification.

The following fact is a folklore.

\begin{proposition}\label{prop:cycle}
Let $G$ be a chordal graph, and let $C$ be a cycle of $G$.
For any vertex $x \in C$, if $x$ is not adjacent to any third vertex of $C$, then the neighbors in $C$ of $x$ are adjacent.
\end{proposition}

We will need a few auxiliary lemmas. The following characterizations of chordal graphs within the class of the $\alpha_1$-metric graphs will be useful.
A graph is said to be an \hnew{\emph{$\alpha_1$-metric}} graph if it satisfies the following: for any $x,y,z,v \in V(G)$ such that $zy\in E(G)$, $z \in I(x,y)$ and $y \in I(z,v)$, $d_G(x,v) \ge d_G(x,y) + d_G(y,v) - 1$ holds~\cite{chepoi:center-triang,yushmanovMetricGraphProperties}.

\begin{lemma} \cite{yushmanovMetricGraphProperties} \label{lem:whenAlphaMetricIsChordal}
$G$ is a chordal graph if and only if it is an $\alpha_1$-metric graph not containing any induced subgraphs isomorphic to cycle $C_5$ and wheel $W_k$, $k \geq 5$.
\end{lemma}

A graph is \emph{bridged}\cite{FARBER1987249} if it contains no isometric cycle $C_k$ of length $k \ge 4$.
Bridged graphs are a natural generalization of chordal graphs. Directly combining two results from~\cite{yushmanovMetricGraphProperties,Dragan2017EccentricityAT}, we obtain the following lemma. 

\begin{lemma} \cite{yushmanovMetricGraphProperties,Dragan2017EccentricityAT} \label{lem:whenBridgedIsAlphaMetric}
$G$ is an $\alpha_1$-metric graph not containing an induced $C_5$ if and only if
$G$ is a bridged graph not containing $W_6^{++}$ as an isometric subgraph (see Figure \ref{fig:w6}).
\end{lemma}

\begin{figure}[h]
\begin{minipage}{\textwidth}
\begin{center}
\includegraphics[scale=1]{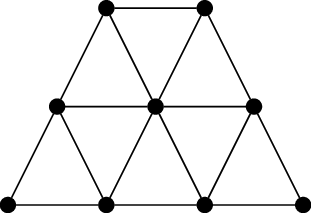}
\captionof{figure}{Forbidden isometric subgraph $W_6^{++}$}\label{fig:w6}
\end{center}
\end{minipage}
\end{figure}

Next lemma establishes conditions in which a Helly graph is chordal.

\begin{lemma} \label{lem:WhenHellyIsChordal}
If $G$ is a Helly graph with no induced wheels $W_k$,  $k \ge 4$, then $G$ is chordal.
\end{lemma}
\begin{proof}
We first claim that $G$ has no induced $C_4$ nor $C_5$.
By contradiction, assume $C_4$ or $C_5$ is induced in $G$.
Consider the system of pairwise intersecting unit disks centered at each vertex of the cycle.
By the Helly property, there is a vertex universal to the cycle.
Thus, $G$ contains $W_4$ or $W_5$, a contradiction establishing the claim that $G$ has no induced $C_4$ nor $C_5$.

We next claim that $G$ is a bridged graph.
Suppose $G$ has an isometric cycle $C_{2\ell}$ for some integer $\ell \geq 3$
(when $\ell=2$, $G$ has induced $C_4$).
Let $x,y \in C_{2\ell}$ be opposite vertices such that $d_G(x,y)=\ell$.
Let $z,t \in C_{2\ell}$ be the distinct neighbors of $y$, as illustrated in \hnew{Figure~\ref{fig:whenHellyIsChordal}(a).}
As the disks $D(x,\ell-2)$, $D(z,1)$, and $D(t,1)$ pairwise intersect,
then by the Helly property, there is a vertex $v \in I(x,t) \cap I(x,z) \cap I(t,z)$.
Since $C_{2\ell}$ is isometric, necessarily $vy \notin E(G)$ and $zt \notin E(G)$.
A contradiction arises with the $C_4$ induced by $v,z,y,t$.

Suppose now that $G$ has an isometric cycle $C_{2\ell+1}$ for some integer $\ell \geq 3$
(when $\ell=2$, $G$ has induced $C_5$).
Let $x,y_1,y_2 \in C_{2\ell+1}$ be vertices such that $y_1y_2 \in E(G)$ and $d_G(x,y_1)=d_G(x,y_2)=\ell$.
As the disks $D(x,\ell-1)$, $D(y_1,1)$, and $D(y_2,1)$ pairwise intersect,
 by the Helly property, there is a vertex $v$ adjacent to $y_1$ and $y_2$ such that $d_G(x,v)=\ell-1$.
Let $z,t \in C_{2\ell+1}$ be vertices such that $z \in N(y_1) \cap I(y_1,x)$ and $t \in N(y_2) \cap I(y_2,x)$, as illustrated in \hnew{Figure~\ref{fig:whenHellyIsChordal}(b).}
Since $\ell \ge 3$ and by choice of the vertices $z,t$ on isometric cycle $C_{2\ell+1}$, necessarily $d_G(z,t) = 3$.
Therefore, $vz \notin E(G)$ or $vt \notin E(G)$; 
without loss of generality, let $vz \notin E(G)$.
As the disks $D(x,\ell-2)$, $D(v,1)$, and $D(z,1)$ pairwise intersect,
by the Helly property, there is a vertex $u \in I(x,v) \cap I(x,z) \cap I(v,z)$.
Necessarily $uy_1 \notin E(G)$, otherwise $d_G(x,y_1) < \ell$.
A contradiction arises with the $C_4$ induced by $u,z,y_1,v$.

Hence, $G$ is a bridged graph. 
Since $G$ has no induced $W_k$ for $k \geq 4$, $G$ does not contain $W_6^{++}$ as an isometric subgraph (observe that $W_6$ is an isometric subgraph of $W_6^{++}$).
By Lemma~\ref{lem:whenBridgedIsAlphaMetric}, $G$ is an $\alpha_1$-metric graph not containing an induced $C_5$.
Since $G$ also has no induced $W_k$ for $k \geq 4$,  by Lemma~\ref{lem:whenAlphaMetricIsChordal}, $G$ is chordal.
\end{proof}
\begin{figure}[h]
\centering
\begin{subfigure}{.4\textwidth}
  \centering
  \includegraphics[height=3cm]{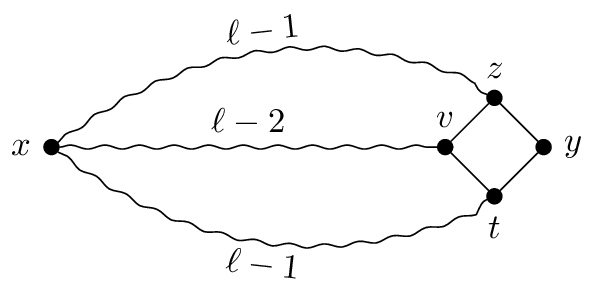}
  \caption{Case of isometric $C_{2\ell}$}
  \label{fig:whenHellyIsChordal1}
\end{subfigure}%
\begin{subfigure}{.5\textwidth}
  \centering
  \includegraphics[height=3cm]{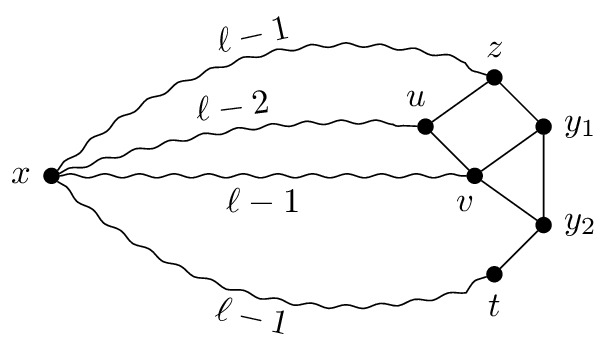}
  \caption{Case of isometric $C_{2\ell+1}$}
  \label{fig:whenHellyIsChordal2}
\end{subfigure}
\caption{Illustration to the proof of Lemma~\ref{lem:WhenHellyIsChordal}.}
\label{fig:whenHellyIsChordal}
\end{figure}

We will also use the fact that chordal graphs and square-chordal graphs can be characterized by the chordality of their so-called visibility graph and intersection graph, respectively.
Let $M = \{S_1, \dots, S_\ell\}$ be a family of subsets of $V(G)$, i.e., each $S_i \subseteq V(G)$.
An \emph{intersection graph} $\calL(M)$ and a \emph{visibility graph} $\Gamma(M)$ are both a generalization of graph powers and are defined by Brandst{\"a}dt et al. \cite{10.1007/11785293_39} as follows.
The sets from $M$ are the vertices of $\calL(M)$ and $\Gamma(M)$.
Two vertices of $\calL(M)$ are joined by an edge if and only if their corresponding sets intersect.
Two vertices of $\Gamma(M)$ are joined by an edge if and only if their corresponding sets are visible to each other; two sets $S_i$ and $S_j$ are visible to each other if $S_i \cap S_j \neq \emptyset$ or there is an edge of $G$ with one end in $S_i$ and the other end in $S_j$.
Denote by $\calD(G) = \{D(v,r) : v \in V(G), r \text{ a non-negative integer}\}$ the family of all disks of $G$.

\begin{lemma} \label{lem:iffGammaIsChordal} \cite{10.1007/11785293_39}
For a graph $G$, $\Gamma(\calD(G))$ is chordal if and only if $G$ is chordal.
\end{lemma}

\begin{lemma} \label{lem:iffIntersectionIsChordal} \cite{10.1007/11785293_39}
For a graph $G$, $\calL(\calD(G))$ is chordal if and only if $G^2$ is chordal.
\end{lemma}

We are now ready to prove the main results of this section.

\begin{theorem} \label{thm:chordalClosedUnderHellification}
Let $G$ be a chordal graph. Then $\calH(G)$ is also chordal.
\end{theorem}
\begin{proof}
By contradiction, assume $G$ is chordal and $H := \calH(G)$ is not.
By Lemma~\ref{lem:WhenHellyIsChordal}, there is an induced wheel $W_k$ in $H$ for some $k \ge 4$.
Let $S=\{v_1,\dots,v_k\}$ be the set of vertices of $W_k$ that induce a cycle $C_k$ suspended by universal vertex $c$.

We first claim that there is a real vertex $u_2$ such that $d_H(u_2,v_i) = d_H(u_2,v_2)+d_H(v_2,v_i)$ for each $v_i \in S$.
Consider the layering $\layer_0,\dots,\layer_\lambda$ produced by a multi-source breadth-first search rooted at the vertex set $\{v_4,v_5,\dots,v_k\}$;
this can be simulated with a BFS rooted at an artificial vertex~$s$ adjacent to only $\{v_4,v_5,\dots,v_k\}$.
Then, $\layer_0 = \{s\}$, $\layer_1 = \{v_4, \dots, v_{k}\}$, $\{v_1, v_3,c\} \subseteq \layer_2$, and $v_2 \in \layer_3$.
\hnew{
Let vertex $u_2$ be a vertex in $\layer_\rho$ such that $\rho$ is maximal and $d_H(v_2,u_2)=\rho - 3$ (i.e., each shortest path from $v_2$ to $u_2$ intersects each layer only once); }
then, $d_H(u_2,v_i) = d_H(u_2,v_2)+d_H(v_2,v_i)$ holds for each $v_i \in S$.
By maximality of $\rho$, there is no vertex $z \in V(H)$ with $I(v_2,u_2) \subset I(v_2,z)$.
Therefore, $u_2$ is a peripheral vertex and, by Proposition~\ref{prop:outsidePointsAreReal}, is real (see Figure~\ref{fig:chordalClosedUnderHellification}).

For each remaining vertex $v_i \in S$, we define a corresponding real vertex $u_i$ in the following way. 
\hnew{By Proposition~\ref{prop:shortestPathSubsetOfReal},
there are two real vertices $u_1,u_3$ such that a shortest path between them contains $P(v_1,v_3)=v_1cv_3$ as a subpath.}
Thus, $d_H(u_1,u_3) = d_H(u_1,v_1) + 2 + d_H(v_3,u_3)$.
Now let $j \in [4,k]$ be an integer.
By choice of $u_2$, 
vertices $c$ and $v_2$ belong to $I(u_2,v_j)$.
Denote by $P(u_2,v_j)$ a shortest path containing $c, v_2$.
By Proposition~\ref{prop:shortestPathSubsetOfReal}, there is a (not necessarily distinct) real vertex $u_j$ such that shortest path $P(u_2,u_j)$ contains $P(u_2,v_j)$.
Thus, $d_H(u_2,u_j) = d_H(u_2,v_2) + 2 + d_H(v_j,u_j)$.

With all distances established, we consider in $G$
the family of disks $\{D(u_i, r(u_i))\}$, where $r(u_i) = d_H(u_i, v_i)$, for each $v_i \in S$.
The disks centered at each vertex $u_i \in V(G)$ are visible to each other if their corresponding vertices $v_i \in V(H)$ are adjacent, i.e.,
$d_G(u_i,u_j) \leq d_H(u_i,v_i) + 1 + d_H(v_j,u_j) = r(u_i) + r(v_i) + 1$ if $v_iv_j \in E(H)$.
As $d_H(u_1, u_3) = r(u_1) + r(u_3) + 2$, the disk $D(u_1,r(u_1))$ and disk $D(u_3,r(u_3))$  are not visible to each other.
As $d_H(u_2, u_j) = r(u_2) + r(u_j) + 2$, for each integer $j \in [4,k]$, the disk $D(u_2,r(u_2))$ is not visible to the disk $D(u_j, r(u_j))$. Consider the visibility graph $\Gamma(\calD(G))$. The vertices $D(u_i,r(u_i)) \in V(\Gamma(\calD(G)))$, $i \in \{1,\dots,k\}$, form a cycle in $\Gamma(\calD(G))$.
As vertex $D(u_2, r(u_2))$ is not adjacent to any vertex $D(u_j, r(u_j))$, where $j \in \{4,\dots,k\}$,
and its neighbors $D(u_1, r(u_1))$ and $D(u_3, r(u_3))$ on the cycle are not adjacent, by Proposition~\ref{prop:cycle}, $\Gamma(\calD(G))$ is not chordal.
By Lemma~\ref{lem:iffGammaIsChordal}, $G$ is also not chordal, a contradiction.
\end{proof}

\begin{figure}[h]
\begin{minipage}{\textwidth}
\begin{center}
\includegraphics[scale=1]{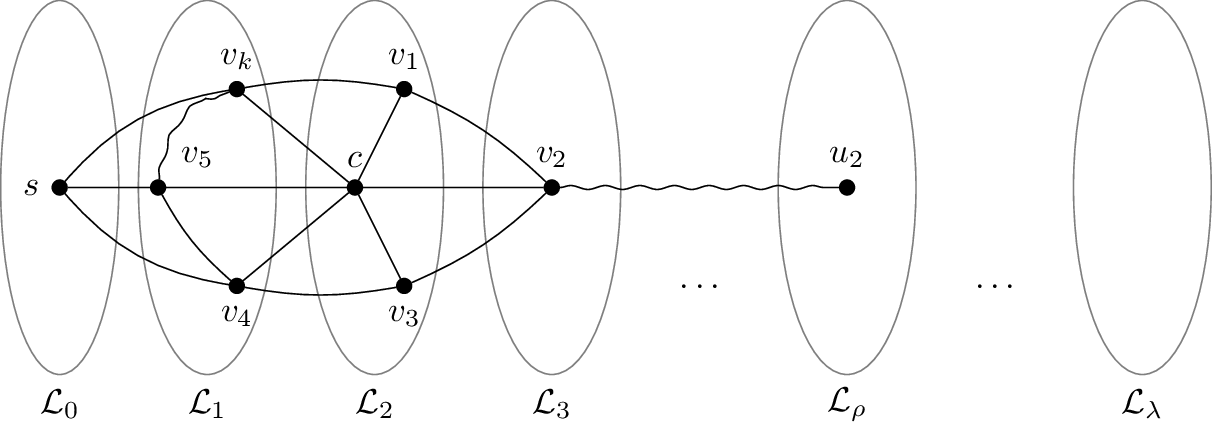}
\captionof{figure}{Illustration to the proof of Theorem~\ref{thm:chordalClosedUnderHellification}.}\label{fig:chordalClosedUnderHellification}
\end{center}
\end{minipage}
\end{figure}

A similar proof shows that the injective hull of a square-chordal graph $G$ is also square-chordal. We will need the following lemma.
\begin{lemma} \cite{FDraganPhD} \label{lem:HellyPowersAreHelly}
Any power of a Helly graph is also a Helly graph.
\end{lemma}

\begin{theorem} \label{thm:squareChordalClosedUnderHellification}
If $G$ is square-chordal, then $\calH(G)$ is square-chordal.
\end{theorem}
\begin{proof}
Let $H := \calH(G)$. By Lemma~\ref{lem:HellyPowersAreHelly}, $H^2$ is Helly.
Assume, by contradiction, that $G^2$ is chordal but $H^2$ is not.
By Lemma~\ref{lem:WhenHellyIsChordal}, there is an induced wheel $W_k$ in $H^2$ for some $k \ge 4$.
Let $S = \{v_1,\dots,v_k \}$ be the set of vertices of $W_k$ that induce a cycle $C_k$ suspended by universal vertex $c$.
As $cv_i \in E(H^2)$ for each $v_i \in S$, then $d_H(c,v_i) \le 2$.
We denote by $v_z$ a particular vertex of $S$ defined as follows.
If there is a vertex $v_i \in S$ such that $cv_i \in E(H)$, then set $v_z := v_i$.
In this case, observe that all vertices $v_j \in S \setminus D_{H^2}(v_i,1)$ satisfy $d_H(v_j,c)=2$, else $S$ would not induce an induced  cycle in $H^2$.
As $k \ge 4$, there is at least one such vertex $v_j$.
On the other hand, if $d_H(c,v_i) = 2$ for each $v_i \in S$, then let $v_z$ be any vertex of $S$. Without loss of generality, in what follows, we can assume that $v_z$ is $v_2$.  

In the next few steps, we define for each vertex $v_i \in S$ a real vertex $u_i$ satisfying particular distance requirements. By Proposition~\ref{prop:shortestPathSubsetOfReal}, there are real vertices $u_1,u_3 \in V(G)$ such that a shortest $(u_1,u_3)$-path in $H$ contains a shortest $(v_1,v_3)$-path in $H$.
As $v_1$ and $v_3$ are non-adjacent in $H^2$, $d_H(v_1,v_3) \ge 3$ and, therefore, $d_H(u_1,u_3) \ge d_H(u_1,v_1) + d_H(v_3,u_3) + 3$.
Consider now a multi-source breadth-first search in $H$ rooted at $M=S \setminus \{v_1,v_2,v_3\}$; this can be simulated with a BFS rooted at an artificial vertex~$s$ adjacent to only the vertices of $M$.
Then, $\layer_0 = \{s\}$, $\layer_1 = M$, $v_1, v_3\in \layer_2 \cup \layer_3$, $c \in \layer_3$, and finally,
\hnew{
$v_2 \in \layer_\mu$ for $\mu = 4$ or $\mu = 5$.
Let vertex $u_2$ be a vertex in $\layer_\rho$ such that $\rho$ is maximal and $d_H(v_2,u_2)=\rho - \mu$ (i.e., each shortest path from $v_2$ to $u_2$ intersects each layer only once); }
By maximality of $\rho$, there is no vertex $f \in V(H)$ with $I(v_2,u_2) \subset I(v_2,f)$.
Therefore, $u_2$ is a peripheral vertex and, by Proposition~\ref{prop:outsidePointsAreReal}, is real.

For each remaining vertex $v_i \in M$, we define a corresponding real vertex $u_i$ in the following way. Note that, by choice of $v_i$, $v_2$ and $v_i$ are non-neighbors in $H^2$; thus, $d_H(v_i,v_2) \ge 3$.
On one hand, if $v_2 \in I(v_i,u_2)$ then, by Proposition~\ref{prop:shortestPathSubsetOfReal}, there is a real $u_i$ vertex such that a shortest $(u_i,u_2)$-path in $H$ contains a shortest $(v_i,v_2)$-path in $H$.
Hence, $d_H(u_i,u_2) \ge d_H(u_i,v_i) + d_H(v_2,u_2) + 3$.
On the other hand, if $v_2 \notin I(v_i,u_2)$, then necessarily $v_2 \in \layer_4$ and $d_H(v_i,v_2)=4$.
Then, there exists a vertex $z \in I(v_i,u_2) \cap \layer_4$ with $d_H(v_i,z)=3$ and $d_H(z,u_2)=d_H(v_2,u_2)$.
By Proposition~\ref{prop:shortestPathSubsetOfReal}, there is a real vertex $u_i$ such that a shortest $(u_i,u_2)$-path in $H$ contains a shortest $(v_i,z)$-path in $H$.
Hence, $d_H(u_i,u_2) = d_H(u_i,v_i) + d_H(v_i,z) + d_H(z,u_2) = d_H(u_i,v_i) + 3 +  d_H(v_i,u_2)$.

With all distances established, we consider in $G$ the family of disks $\{D(u_i, r(u_i))\}$, where $r(u_i) = d_H(u_i, v_i) + 1$, for each $v_i \in S$.
The disks centered at each vertex $u_i \in V(G)$ intersect if their corresponding vertices $v_i \in V(H^2)$ are adjacent in $H^2$, i.e.,
$d_G(u_i,u_j) \leq d_H(u_i,v_i) + 2 + d_H(v_j,u_j) = r(u_i) + r(u_j)$ if $v_iv_j \in E(H^2)$.
As $d_H(u_1,u_3) \ge r(u_1) + r(u_3) + 1$, the disk $D(u_1,r(u_1))$ and disk $D(u_3,r(u_3))$ do not intersect.
As $d_H(u_2,u_j) \ge r(u_2) + r(u_j) + 1$, for each $j \in \{4,\dots,k\}$, the disks $D(u_2,r(u_2))$ and $D(u_j,r(u_j))$ do not intersect. Consider the intersection graph $\calL(\calD(G))$. The vertices $D(u_i,r(u_i)) \in V(\calL(\calD(G)))$, $i \in \{1,\dots,k\}$, form a cycle in $\calL(\calD(G))$.
As vertex $D(u_2,r(u_2))$ is not adjacent to any vertex $D(u_j,r(u_j))$, where $j \in \{4,\dots,k\}$, 
and its neighbors $D(u_1,r(u_1))$ and $D(u_3,r(u_3))$ on the cycle are not adjacent,
 by Proposition~\ref{prop:cycle}, $\calL(\calD(G))$ is not chordal.
By Lemma~\ref{lem:iffIntersectionIsChordal}, $G^2$ is also not chordal, a contradiction.
\end{proof}

A graph $G$ is \emph{dually chordal} if it has a so-called  \emph{maximum neighborhood ordering} (see~\cite{draganLocationProblems,DBLP:journals/siamdm/BrandstadtDCV98} for definitions and various characterizations of this class of graphs). A maximum neighborhood ordering can be constructed in total linear time~\cite{draganLocationProblems,DBLP:journals/siamdm/BrandstadtDCV98}. For us here, the following characterization is relevant: a graph $G$ is dually chordal if and only if $G$ is neighborhood-Helly and $G^2$ is chordal~\cite{draganLocationProblems,DBLP:journals/siamdm/BrandstadtDCV98}.  So, we can state the following corollary. 

\begin{corollary} \label{cor:gSquaredDuallyChordal} 
If $G$ is a square-chordal graph, then $\calH(G)$ is dually chordal.
\end{corollary}


\section{Distance-hereditary graphs}\label{sec:dhg}
A graph is \emph{distance-hereditary} if and only if each of its connected induced subgraphs is isometric~\cite{howorkaDHG}, that is,
the length of any induced path between two vertices equals their distance in~$G$.
In this section, we show that distance-hereditary graphs are closed under Hellification.
We give a characterization of the distance-hereditary Helly graphs \anew{and show conditions under which adding a vertex to a Helly graph keeps it Helly in Section~\ref{sec:dhHelly}.
In Section~\ref{sec:dhDataStruc} we describe a data structure which we then use in Section~\ref{sec:dhCompH}} to construct the injective hull of a distance-hereditary graph in linear time.

We use the following characterizations of distance-hereditary graphs.
\begin{proposition}\label{prop:dhgCharacterization}  \cite{BANDELT1986182,doi:10.1137/0217032} 
For a graph $G$, the following conditions are equivalent:
\setlist{nolistsep}
\begin{enumerate}[noitemsep, label=(\roman*)]
	\item $G$ is distance-hereditary;
	\item\label{byForbiddenSubgraphs} The house, domino, gem, and the cycles $C_k$ of length $k \geq 5$ are not induced subgraphs of $G$ (see Figure  \ref{fig:dhg});
	\item\label{byPruningSequence}
	$G$ is obtained from $K_1$ by a sequence of one-vertex extensions: attaching a pendant vertex or a twin vertex.
\end{enumerate}
\end{proposition}

\begin{figure}
    [htb] 
     \vspace*{-.2cm}
    \centering
    \includegraphics[]{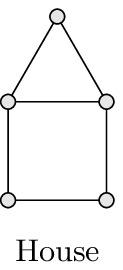}%
    \hspace*{1cm}%
    \includegraphics[]{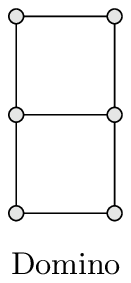}%
    \hspace*{1cm}%
    \includegraphics[]{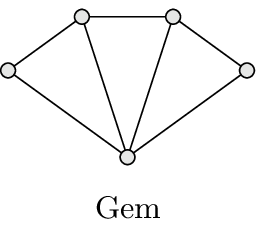}%
    \caption
    {
        Forbidden induced subgraphs in a distance-hereditary graph.
    }
    \label{fig:dhg} %
\end{figure}

\anew{\subsection{Helly property for distance-hereditary graphs}
\label{sec:dhHelly}}

Let $w,x,y,z$ be four vertices that induce a $C_4$. We denote by $S(w,x,y,z)$ 
an extended square that includes the vertices $w,x,y,z$ which induce a $C_4$ and any vertex adjacent to at least three of them, i.e., $S(w,x,y,z) = \{ v \in D(\{w,x,y,z\},1) : |N[v] \cap \{w,x,y,z\}| \geq 3 \}$.
We show that a distance-hereditary graph is Helly if and only if all extended squares are suspended.
The result is analogous to a characterization of chordal Helly graphs~\cite{FDraganPhD}: a chordal graph is Helly if and only if all extended triangles are suspended, where an extended triangle $S(x,y,z)$ is defined as the set of vertices that see at least two vertices of the triangle $\Delta(x,y,z)$.

\begin{lemma}\label{lem:dh-helly}
 A distance-hereditary graph $G$ is Helly if and only if, for every $C_4$ induced by $w,x,y,z \in V(G)$, the extended square $S(w,x,y,z)$ is suspended.
\end{lemma}
\begin{proof}
It is known that distance-hereditary graphs are pseudo-modular~\cite{Bandelt:1986:PG:10100.10102}. Hence, by Proposition~\ref{prop:pseudo-modular-helly}, $G$ is Helly if and only if it is neighborhood-Helly, i.e., all 2-sets are suspended.
  If $G$ is neighborhood-Helly, then all extended squares are suspended since all vertices of an extended square are pairwise at distance at most 2.
  We assert that if every extended square is suspended, then all 2-sets are suspended.

  We use an induction on the cardinality of 
  a 2-set.
  Assume any 2-set $M \subseteq V(G)$ with $|M| \leq k$ is suspended.
  Clearly, it is true for $k \le 2$.
  By contradiction, assume every extended square is suspended but there is an unsuspended 2-set $M$ with $|M|=k+1 \ge 3$. 
  Let $v_1, v_2, v_3 \in M$.
  By the inductive hypothesis,
  there is a vertex $c_1$ universal to $M \setminus \{v_1\}$, a vertex $c_2$ universal to $M \setminus \{v_2\}$, and a vertex $c_3$ universal to $M \setminus \{v_3\}$.
  Since $M$ is not suspended, necessarily each $c_i \in \{c_1,c_2,c_3\}$ has $c_iv_i \notin E(G)$, $c_i \neq v_i$, and $c_i$ is distinct from the other two vertices of $\{c_1,c_2,c_3\}$.
  We consider three cases based on how many of $\{c_1,c_2,c_3\}$ are distinct from the three vertices $\{v_1,v_2,v_3\}$.
  We will obtain a forbidden induced subgraph which contradicts Proposition~\ref{prop:dhgCharacterization}\ref{byForbiddenSubgraphs} or will show that $M$ is a subset of some extended square which is suspended, giving a contradiction with $M$ being unsuspended.

  {\em Case 1. $c_3 \notin \{v_1,v_2,v_3\}$ and $c_1,c_2 \in \{v_1,v_2,v_3\}$.}\\
  	Without loss of generality, let $c_2 = v_1$. 
    If $c_1=v_3$, then $v_3c_2 \in E(G)$, i.e., $c_1v_1 \in E(G)$, a contradiction.
    Hence, $c_1=v_2$ and therefore, $c_1,c_3,c_2,v_3$ induce $C_4$.
    Any $x \in M$ belongs to $S(c_1,c_3,c_2,v_3)$ since $x$ is either one of $c_1,c_2,v_3$ or adjacent to all of $c_1,c_2,c_3$.
    Thus, $M \subseteq S(c_1,c_3,c_2,v_3)$.

  {\em Case 2. $c_2,c_3 \notin \{v_1,v_2,v_3\}$ and $c_1 \in \{v_1,v_2,v_3\}$.}\\
   Without loss of generality, let $c_1=v_2$.
    Then, $c_1$ is adjacent to $v_3$, but $c_1v_1 \notin E(G)$.
    Since $c_3v_2 \in E(G)$ and $c_2v_2 \notin E(G)$, by equality $c_3c_1 \in E(G)$ and $c_2c_1 \notin E(G)$.
    By assumption, $c_2$ is adjacent to $v_3$ and $v_1$, $c_3$ is adjacent to $v_1$, and $c_3v_3 \notin E(G)$.
    It only remains whether $c_2c_3 \in E(G)$ and/or $v_1v_3 \in E(G)$.
    If at most one of those edges occurs, we obtain $C_5$ or a house induced by $\{c_1,c_3,v_1,c_2,v_3\}$.
    Therefore, both edges $c_2c_3$ and $v_1v_3$ must be present. 
    Now, any $x \in M \setminus \{c_2,v_3,v_1\}$ is adjacent to both $c_1,c_3$.
    Furthermore, if $xv_1 \notin E(G)$ and $xv_3 \notin E(G)$, then $x,c_3,c_1,v_3,v_1$ induce a house.
    Thus, $x$ is also adjacent to at least one of $v_1,v_3$.
    Since any $x \in M$ is either one of $v_1,v_3$ or is  adjacent to at least three of $\{c_1,c_3,v_1,v_3\}$, we get $M \subseteq S(c_1,c_3,v_1,v_3)$.

  {\em Case 3. $c_1,c_2,c_3 \notin \{v_1,v_2,v_3\}$.}\\
  	By assumption, $c_1$ is adjacent to $v_2$ and $v_3$,
  	$c_2$ is adjacent to $v_1$ and $v_3$,
  	$c_3$ is adjacent to $v_1$ and $v_2$,
  	and  $c_1v_1,c_2v_2,c_3v_3 \notin E(G)$.
  	If each $i,j \in \{1,2,3\}$ with $i \neq j$ satisfies $c_ic_j \notin E(G)$ and $v_iv_j \notin E(G)$, then $v_1,c_3,v_2,c_1,v_3,c_2$ induce $C_6$.
  	Thus, there is some chord $c_ic_j \in E(G)$ or $v_iv_j \in E(G)$. We consider two subcases without loss of generality.

  	{\em Case 3(a). There is a chord $c_2c_3 \in E(G)$.} \\
  	If there are no other edges between vertices $\{c_2,c_3,v_2,c_1,v_3\}$, then those vertices induce a $C_5$. Thus, there is at least one of the following chords: $c_1c_2$, $v_3v_2$, or $c_1c_3$.
	If $v_2v_3 \notin E(G)$, we get in $G$ a house or gem induced by $v_3,c_2,c_3,v_2,c_1$. Hence, $v_2v_3 \in E(G)$. Consider now $C_4$ induced by $c_2,c_3,v_2,v_3$. Any vertex $x \in M \setminus \{c_2,c_3,v_2,v_3\}$ is adjacent to both $c_2,c_3$. Furthermore, if $xv_2 \notin E(G)$ and $xv_3 \notin E(G)$, then $x,c_3,c_2,v_3,v_2$ induce a house. Thus, $x$ is also adjacent to at least one of $v_2,v_3$. Since any $x \in M$ is either one of $c_2,c_3,v_2,v_3$ or is  adjacent to at least three of $\{c_2,c_3,v_2,v_3\}$, we get $M \subseteq S(c_2,c_3,v_2,v_3)$.
\commentout{    
      	We first claim that $c_1c_2 \notin E(G)$. By contradiction, assume the edge exists.
  	If $v_2v_3 \notin E(G)$, we get in $G$ a house or gem induced by $v_3,c_2,c_3,v_2,c_1$.
  	Hence, $v_2v_3 \in E(G)$.
  	If $c_1c_3 \notin E(G)$ and $v_1v_2 \notin E(G)$, we get in $G$ a house induced by $c_1,c_2,v_1,c_3,v_2$. 
  	If $c_1c_3 \notin E(G)$ and $v_1v_2 \in E(G)$, then $M \subseteq S(c_2,c_1,c_2,c_3)$.\todo{Why?}
  	Hence, $c_1c_3 \in E(G)$.
  	Then, any $x \in M \setminus \{v_2,v_3\}$ is adjacent to $c_2$ and $c_3$.
  	Moreover, if $xv_3 \notin E(G)$ and $xv_2 \notin E(G)$ then we get in $G$ a house induced by $v_3c_2xc_3v_2$.
  	Hence, $x$ is also adjacent to at least one of $v_2,v_3$, and so $M \subseteq S(v_3,c_2,c_3,v_2)$, a contradiction that establishes $c_1c_2 \notin E(G)$.

  	We next claim that $c_1c_3 \notin E(G)$. By contradiction, assume the edge exists.
  	If $v_3v_2 \notin E(G)$, then we get in $G$ a house induced by $v_2c_1v_3c_2c_3$.
  	If $v_3v_2 \in E(G)$, then we get in $G$ a house induced by $c_1v_3c_2v_1c_3$ or $M \subseteq S(v_3,c_1,c_2,c_3)$, a contradiction that establishes $c_1c_3 \notin E(G)$.

  	Necessarily, we have the chord $v_3v_2 \in E(G)$ and contradiction arises with house induced by $c_1v_3c_2c_3v_2$.
} 

    {\em Case 3(b). There is a chord $v_2v_3 \in E(G)$ and $c_ic_j \notin E(G)$ for distinct $i,j \in \{1,2,3\}$.}\\
    If there are no other edges between vertices  $\{v_3,c_2,v_1,c_3,v_2\}$, then those vertices induce $C_5$.
    Thus, there is at least one of the following chords: $v_1v_3$ or $v_1v_2$.
    But then, vertices $v_2,v_3,c_2,v_1,c_3$ induce a house or a gem. 
    Obtained contradictions prove the lemma.
\end{proof}

We found it advantageous to use a characteristic pruning sequence of $G$ (see Proposition \ref{prop:dhgCharacterization}(iii)).
A \emph{pruning sequence} $\sigma_G : V(G) \rightarrow \{1,\dots,n\}$ of $G$ is a total ordering of its vertex set $V(G) = \{v_1,\dots,v_n\}$ such that each vertex $v_i$ satisfies one of the following conditions in the induced subgraph $G_i:=\langle v_1, \dots, v_i \rangle$:
\begin{enumerate}[label=(\roman*), noitemsep, nolistsep]
    \item $v_i$ is a pendant vertex to some vertex~$v_j$ with $\sigma_G(v_j) < \sigma_G(v_i)$,
    \item $v_i$ is a true twin of some vertex~$v_j$ with $\sigma_G(v_j) < \sigma_G(v_i)$, or
    \item $v_i$ is a false twin of some vertex~$v_j$ with $\sigma_G(v_j) < \sigma_G(v_i)$.
\end{enumerate}


Next lemmas give conditions under which adding a vertex to a Helly graph keeps it Helly. 
Consider a graph $H$ obtained by adding to a Helly graph $G$ a vertex $u$ as a pendant or twin to some vertex in $G$.
We show that any family $\calF = \big \{ D_H(w, r(w)) : w \in M \subseteq V(H) \big \}$ of pairwise intersecting disks in $H$ has a common intersection.
Note that this is trivially true if any vertex $w \in M$ has $r(w)=0$ (since $w$ is common to all disks of $\calF$)
or if $u \notin M$ (since $G$ is isometric in $H$ and the family of pairwise intersecting disks $\big \{ D_G(w, r(w)) : w \in M \subseteq V(G) \big \}$ have a common intersection in $G$).

\begin{lemma} \label{lem:pendantHelly}
Let $G + \{u\}$ be a graph obtained by adding a vertex $u$ pendant to $v \in V(G)$.
If $G$ is Helly, then $G+\{u\}$ is Helly.
\end{lemma}
\begin{proof}
Let $H:= G + \{u\}$,  $\calF = \big \{ D_H(w, r(w)) : w \in M \subseteq V(H) \big \}$ be a family of pairwise intersecting disks in $H$, and $u \in M$.
If $r(u) \geq 2$, one may substitute in $\calF$ the disk  $D_H(u, r(u))$ with the equivalent disk~$D_H(v, r(u)-1)$. Since $G$ is isometric in $H$ and is Helly, the corresponding disks in $G$ have a common intersection.
Assume now that $r(u) = 1$.
Then, $v \in D_H(u, r(u))$. 
As the disks of $\calF$ pairwise intersect, every $w \in M \setminus \{u\}$ satisfies $r(w)+r(u) \ge d_H(w,u)=d_H(w,v)+1=d_H(w,v)+r(u)$.
Hence, $d_H(w,v) \le r(w)$ and vertex $v$ is common to all disks.
\end{proof}

\begin{lemma} \label{lem:trueTwinHelly}
Let $G+\{u\}$ be a graph obtained by adding a vertex $u$ as a true twin to $v \in V(G)$.
If $G$ is Helly, then $G+\{u\}$ is Helly.
\end{lemma}
\begin{proof}
Let $H:= G + \{u\}$,  $\calF = \big \{ D_H(w, r(w)) : w \in M \subseteq V(H) \big \}$ be a family of pairwise intersecting disks in $H$, and $u \in M$.
Because $u$ is a true twin of~$v$, $D_H(u,r)=D_H(v,r)$ for any radius $r \ge 1$.
Hence, we can assume that each disk of $\calF$ is centered at a vertex of $G$, therefore there is a common intersection of all disks of $\calF$.
\end{proof}

\begin{lemma} \label{lem:falseTrueTwinHelly}
Let $G + \{ u \}$ be the graph obtained by adding a vertex $u$ as a false twin to $v \in V(G)$.
If $G$ is Helly and \anew{there is some $y \in V(G)$ with $N[v] \subseteq N[y]$}, then $G+\{u\}$ is Helly.
\end{lemma}

\begin{proof}
Let $H := G + \{ u \}$, $\calF = \big \{ D_H(w, r(w)) : w \in M \subseteq V(H) \big \}$ be a family of pairwise intersecting disks in $H$, and $u \in M$.
If $r(u) > 1$, then $D_H(v, r(u)) = D_H(u, r(u))$ and so we can assume that each disk of $\calF$ is centered at a vertex of $G$, implying a non-empty common intersection of all disks of $\calF$.
Assume now that $r(u) = 1$.
As each disk of $\calF' = \bigl( \calF \setminus \{ N[u] \} \bigr) \cup \{ N[v] \}$ is centered at a vertex of $G$, there is a common intersection $R$ of all disks of $\calF'$.
Since $N[v] \in \calF'$, $R \subseteq N[v]$.
\anew{Recall that there is a vertex $y$ in $G$ with $N[v] \subseteq N[y]$.}
Therefore, $v \in R$ implies $y \in R$.
Thus, there exists vertex $s \in R \cap N(v)$.
By definition of $u$, $N(u) = N(v)$.
Thus, $s \in N(u)$ and hence $s$ is contained in each disk in $\calF$.
\end{proof}



\anew{
\subsection{A data structure for subsets of neighborhoods}
\label{sec:dhDataStruc}

The idea for our data structure is based on the partition refinement data structure.
It was introduced in \cite{PaigeTarjan1987} and allows to find all twins in a graph in linear time.%
\footnote{To do so, start with the set $V(G)$ and, for each vertex $v$, call Refine($N[v]$) for true twins or Refine($N(v)$) for false twins.}
Similar to a partition refinement, our data structure handles sets of vertices.
In addition to that, it also adds directed edges between sets.
We create these sets and edges in such a way that the following two properties are satisfied:
two vertices are in the same set if and only if they are true twins, and there is an edge from a set $X$ to a set $Y$ if and only if $N[x] \subset N[y]$ for each $x \in X$ and $y \in Y$.

To construct our data structure, we use a pruning sequence $(v_1, \ldots, v_n)$ of a given graph $G$.
Let $G_i$ denote the graph induced by $\{ v_1, \ldots, v_i \}$ and let $N_i[v]$ denote the closed neighborhood of a vertex $v$ with respect to $G_i$.
For $G_1$, our data structure only contains a single set $S = \{ v_1 \}$.
Each time we add a vertex to $G_i$, we update our data structure as follows to ensure both properties are still satisfied.

Assume that we have three sets $S$, $X$, and $Y$ with $u, v \in S$, $x \in X$, and $y \in Y$.
Additionally, let there be an edge from $X$ to $S$ and from $S$ to $Y$.
Hence, $N_i[x] \subset N_i[u] = N_i[v] \subset N_i[y]$.
Let $G_{i + 1}$ be the graph created by adding a vertex $w$.
We now have three cases:
(i)~if $w$ is pendant to $v$, then create two new sets $S_v := \{ v \}$ and $S_w := \{ w \}$, set $S := S \setminus \{ v \}$, and add the edges $XS_v$, $SS_v$, and $S_wS_v$;
(ii)~if $w$ is a true twin of $v$, set $S := S \cup \{ w \}$; and
(iii)~if $w$ is false twin of $v$, create two new sets $S_v := \{ v \}$ and $S_w := \{ w \}$, set $S := S \setminus \{ v \}$, and add the edges $S_vS$, $S_vY$, $S_wS$, and $S_wY$.
Note that, $X$ and $Y$ are not necessarily unique.
Hence, when adding an edge from $X$ or to $Y$, we have to add such an edge for each such set $X$ and $Y$.
See Figure~\ref{fig:dhDataStructure} for an illustration.

\begin{figure}
  \centering
  \begin{subfigure}{.5\textwidth}
    \centering
    \includegraphics[]{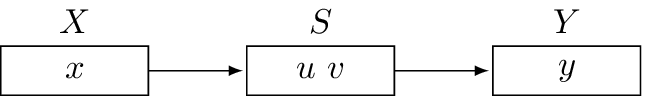}
    \caption{Data structure before adding $w$.}
  \end{subfigure}%
  \begin{subfigure}{.5\textwidth}
    \centering
    \includegraphics[]{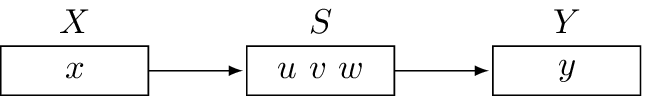}
    \caption{Data structure after adding true twin $w$ of $v$.}
  \end{subfigure}

  \bigskip

  \begin{subfigure}{.5\textwidth}
    \centering
    \includegraphics[]{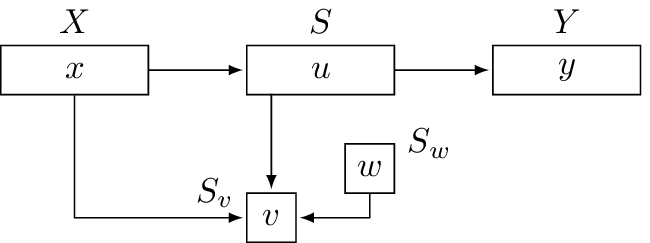}
    \caption{Data structure after adding $w$ pendant to $v$.}
  \end{subfigure}%
  \begin{subfigure}{.5\textwidth}
    \centering
    \includegraphics[]{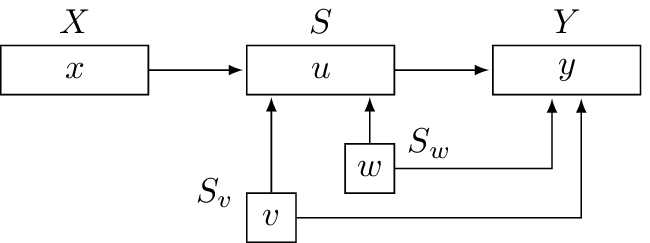}
    \caption{Data structure after adding false twin $w$ of $v$.}
  \end{subfigure}
  \caption{Modification of the data structure when adding a vertex $w$ to $G_i$.}
  \label{fig:dhDataStructure}
\end{figure}

There is a special case for the second vertex $v_2$ which only happens for that vertex.
After adding $v_2$, it is a true twin \emph{and} a pendent vertex to $v_1$.
The construction of our data structure above, however, assumes that, if $w$ is pendent to $v$, then $v$ has some neighbor not adjacent to $w$.
Therefore, to construct the data structure correctly, $v_2$ should be treated as true twin of $v_1$ and not as pendant vertex.

\begin{lemma}
\label{lem:dhDS_sets}
In the data structure constructed above, two vertices are in the same set if and only if they are true twins.
\end{lemma}

\begin{proof}
The lemma is clearly satisfied for $G_1$.
Assume now, by induction, that our data structure satisfies Lemma~\ref{lem:dhDS_sets} for $G_i$ and let $A \neq S$ be a set handled by the data structure.
After adding $w$, either all vertices in $A$ are adjacent to $w$ (if $w$ is a twin of $v$ and vertices in $A$ are adjacent to $v$) or non of them are.
All other neighbors remain the same.
Hence, all sets $A \neq S$ still satisfy Lemma~\ref{lem:dhDS_sets} after adding $w$.
To analyse $S$, we need to distinguish between the three cases of $w$.

\emph{Case~(i): $w$ is pendant to $v$.}
After adding $w$, we have $N_{i + 1}[u] = N_i[u]$ for each $u \in S$, $N_{i + 1}[v] = N_i[v] \cup \{ w \}$ and $N_{i + 1}[w] = \{ v, w \}$.
It follows that $v$ and $w$ have no true twins in $G_{i + 1}$ and all remaining vertices in $S$ (with respect to $G_{i + 1}$) are still true twins.
Hence, by placing $v$ and $w$ into their own respective sets, the data structure still satisfies Lemma~\ref{lem:dhDS_sets}.

\emph{Case~(ii): $w$ is a true twin of $v$.}
In this case, clearly, $N_{i + 1}[u] = N_{i + 1}[v] = N_{i + 1}[w]$ for each $u \in S$.
Hence, by adding $w$ to $S$, the data structure still satisfies Lemma~\ref{lem:dhDS_sets}.

\emph{Case~(iii): $w$ is a false twin of $v$.}
After adding $w$, we have $N_{i + 1}[u] = N_i[u] \cup \{ w \}$ for each $u \in S$, $N_{i + 1}[v] = N_i[v]$, and $N_{i + 1}[w] = N_i(v) \cup \{ w \}$.
It follows that $v$ and $w$ have no true twins in $G_{i + 1}$ and all remaining vertices in $S$ (with respect to $G_{i + 1}$) are still true twins.
Hence, by placing $v$ and $w$ into their own respective sets, the data structure still satisfies Lemma~\ref{lem:dhDS_sets}.
\end{proof}

\begin{lemma}
\label{lem:dhDS_edges}
In the data structure constructed above, there is an edge from a set $A$ to a set $B$ if and only if $N[a] \subset N[b]$ for each $a \in A$ and $b \in B$.
\end{lemma}

\begin{proof}
The lemma is clearly satisfied for $G_1$.
Assume now, by induction, that our data structure satisfies Lemma~\ref{lem:dhDS_edges} for $G_i$ and let $a$ and $b$ be two vertices in $G_i$ with $N_i[a] \nsubseteq N_i[b]$.
Clearly, since we only add a new vertex and do not remove any existing vertices, $N_{i + 1}[a] \nsubseteq N_{i + 1}[b]$.
Now assume that $a, b \notin \{ v, w \}$ and $N_i[a] \subseteq N_i[b]$.
If $a$ is adjacent to $w$ in $G_{i + 1}$, then $a$ is adjacent to $v$ in $G_i$.
It follows that $b$ is adjacent to $v$ and, hence, $w$ too.
Therefore, $N_i[a] \subseteq N_i[b]$ implies $N_{i + 1}[a] \subseteq N_{i + 1}[b]$ for all $a, b \notin \{ v, w \}$.
Lemma~\ref{lem:dhDS_edges} is therefore satisfied for each pair of sets $A, B \neq S$.
To analyse the edges of $S$, $S_v$, and $S_w$, we need to distinguish between the three cases of $w$.

\emph{Case~(i): $w$ is pendant to $v$.}
After adding $w$, we have $N_{i + 1}[a] = N_i[a]$ for each $a \notin \{ v, w \}$, $N_{i + 1}[v] = N_i[v] \cup \{ w \}$ and $N_{i + 1}[w] = \{ v, w \}$.
It follows that the added edges $XS_v$, $SS_v$, and $S_wS_v$ are needed to satisfy Lemma~\ref{lem:dhDS_edges}, and that adding any other edge would violate Lemma~\ref{lem:dhDS_edges}.

\emph{Case~(ii): $w$ is a true twin of $v$.}
In this case, clearly, $N_{i + 1}[u] = N_{i + 1}[v] = N_{i + 1}[w]$ for each $u \in S$.
Additionally, for each vertex $a$, $N_i[a] \subseteq N_i[v]$ if and only if $N_{i + 1}[a] \subseteq N_{i + 1}[v]$, and $N_i[v] \subseteq N_i[a]$ if and only if $N_{i + 1}[v] \subseteq N_{i + 1}[a]$.
Hence, the data structure still satisfies the lemma after adding $w$ into $S$.

\emph{Case~(iii): $w$ is a false twin of $v$.}
After adding $w$, we have
$N_{i + 1}[u] = N_i[u] \cup \{ w \}$ for each $u \in S$,
$N_{i + 1}[v] = N_i[v]$,
$N_{i + 1}[w] = N_i(v) \cup \{ w \}$,
$N_{i + 1}[x] = N_i[x] \cup \{ w \}$, and
$N_{i + 1}[y] = N_i[y] \cup \{ w \}$.
It follows that the added edges $S_vS$, $S_vY$, $S_wS$ and $S_wY$ are needed to satisfy Lemma~\ref{lem:dhDS_edges}, and that adding any other edge would violate Lemma~\ref{lem:dhDS_edges}.
\end{proof}

Before discussing the efficiency of our data structure, observe the following.
If $w$ is a pendant vertex or false twin of $v$, we remove $v$ from $S$.
It can therefore happen that $S$ becomes empty.
In that case, instead of removing $v$ from $S$ and creating a new set $S_v$, we leave $v$ in $S$, $S$ becomes $S_v$, and we update the edges accordingly.
That is, we remove all outgoing edges if $w$ is pendant to $v$, or we remove all incoming edges if $w$ is a false twin of $v$.

\begin{lemma}
\label{lem:dhDS_time}
For a given distance-hereditary graph $G$ and a corresponding pruning sequence, the overall runtime to construct the data structure as described above is at most linear with respect to the size of $G$.
\end{lemma}

\begin{proof}
When constructing the data structure, we only add edges if a new set is created, an edge between two sets is created and removed at most once, and there is at least one edge between two vertices in $G$ for each edge between two sets.
Additionally, each set is created at most once and contains at least one vertex.
Therefore, the overall size of the data structure is at most as large as the size of $G$ and the runtime to construct it is at most linear.
\end{proof}
}

\anew{\subsection{Computing the injective hull}
\label{sec:dhCompH}}

We next show that one can efficiently compute the injective hull of a distance-hereditary graph~$G$.
Moreover, as a byproduct, we get that $\calH(G)$ is distance-hereditary and $|V(\calH(G))| \in O(|V(G)|)$.
%
One attempt to compute $\calH(G)$ is to add Helly vertices suspending all maximal 2-sets of $G$.
We observe that $G$ has $O(n)$  maximal 2-sets.
Indeed, since $G^2$ is chordal~\cite{BANDELT199537}, there are $O(n)$ maximal cliques in $G^2$ obtainable via a perfect elimination ordering 
of $G^2$, and there is one-to-one correspondence between maximal cliques of $G^2$ and maximal 2-set of $G$.
Since adding a Helly vertex $h$ to suspend a single 2-set in $G$ may create another unsuspended 2-set in $G+\{h\}$, this information alone is insufficient to conclude but gave a promising indication that possibly $|V(\calH(G))| \in O(|V(G)|)$.
Second attempt based on Lemma~\ref{lem:dh-helly} is to suspend all extended squares, which incurs a similar problem that $G+\{h\}$ may have a new unsuspended extended square. Additionally, there can be more extended squares than there are maximal 2-sets.

Using Lemma~\ref{lem:pendantHelly}, Lemma~\ref{lem:trueTwinHelly}, Lemma~\ref{lem:falseTrueTwinHelly}, a pruning sequence of a distance-hereditary graph $G$, \anew{and the data structure described in the previous subsection, we can compute $\calH(G)$ in linear time.} Moreover, as a byproduct, we get that $\calH(G)$ is distance-hereditary, too.
\begin{theorem}
If $G$ is a distance-hereditary graph, then $\calH(G)$ is distance-hereditary and can be computed in $O(n + m)$ time, where $n = |V(G)|$ and $m = |E(G)|$.
\end{theorem}

\begin{proof}
We use a pruning sequence $\sigma_G=(v_1,\dots,v_n)$ of $G$.
Let $G_i$ denote the graph induced by $\{v_1,\dots,v_i\}$ and let $H$ be a graph that initially contains only $v_1$; clearly $H$ is Helly, distance-hereditary, and contains $G_1$ as an isometric subgraph.
We iterate over the remaining vertices $v_i \in \sigma_G$, $i \ge 2$, to carefully attach new vertices to $H$ as pendants/twins to old vertices in $H$, thereby constructing for $H$ a pruning sequence $\sigma_H$.
Then, by Proposition~\ref{prop:dhgCharacterization}\ref{byPruningSequence}, $H$ is distance-hereditary.
\anew{Additionally, we maintain a data structure for $H$ as described in Section~\ref{sec:dhDataStruc} above.}
For clarity, denote by $H_k$ the graph induced by $\{u \in V(H) : \sigma_H(u) \le k\}$.
We claim that the resulting graph \anew{$H_k = \calH(G)$}, where $k={|V(H)|}$.
There are two cases.

{\em Case 1. Next vertex $v_i \in \sigma_G$ is a pendant or true twin to some vertex $v_j$ in $G_i$.}
Set \anew{$H := H + \{v_i\}$} (the graph obtained by adding $v_i$ as a pendant or true twin to $v_j$ in $H$).
By Lemma~\ref{lem:pendantHelly}, and Lemma~\ref{lem:trueTwinHelly}, $H$ remains Helly.
As $H$ is distance-hereditary and contains $G_i$ as an induced subgraph, $G_i$ is isometric in $H$.

{\em Case 2. Next vertex $v_i \in \sigma_G$ is a false twin to some vertex $v_j$ in $G_i$.}
\anew{%
Use the data structure for $H$ to determine if $H$ contains a vertex $y \neq v_j$ with $N[v_j] \subseteq N[y]$ as follows.
Let $S$ be the set containing $v_j$.
By Lemma~\ref{lem:dhDS_sets} and Lemma~\ref{lem:dhDS_edges}, such a $y$ exits if and only if $|S| > 1$ or $S$ has an outgoing edge.
If $H$ contains no such $y$, then we first create a new true twin $y$ of $v_j$ in $H$ and set $H := H + \{ y \}$.
Next, set $H := H + \{v_i\}$ (the graph obtained by adding $v_i$ as a false twin to $v_j$ in $H$).
By Lemma~\ref{lem:falseTrueTwinHelly} $H$ remains Helly.
As $H$ is distance-hereditary and contains $G_i$ as an induced subgraph, $G_i$ is isometric in $H$.

A pruning sequence for $G$ can be constructed in linear time~\cite{BANDELT1986182,doi:10.1137/0217032,DAMIAND200199}.
By Lemma~\ref{lem:dhDS_time}, the data structure for $H$ can be constructed in linear time, too.
Checking if $H$ contains a vertex $y \neq v_j$ with $N[v_j] \subseteq N[y]$ (case~2) can then be done in constant time.
Note that, for each vertex $v_i \in V(G)$, there is at most one vertex $y_i$ added to $H$.
Similarly, for each edge $v_iv_j$ where $v_i$ is a false twin to $u$ in $G_i$, we add at most three additional edges ($y_iv_i$, $y_iv_j$, and $y_iu$) to $H$.
Therefore, $H$ has at most $2n$ vertices and $4m$ edges and can be constructed in $O(n + m)$ time.
}

We finally claim that $H$ is a \emph{minimal} Helly graph that contains $G$ as an isometric subgraph, i.e., $H = \calH(G)$.
By contradiction, assume there is vertex $y \in V(H) \setminus V(G)$ with minimal $\sigma_H(y)$ such that $H \setminus \{y\}$ is Helly.
Since $y \notin V(G)$, by algorithm construction, there is a vertex $v_i \in V(G_i)$ which is a false twin to $v_j$ in $G_i$, but no vertex in $V(H_{\sigma_H(v_i)}) \setminus \{y\}$ suspends the 2-set $M= D_{G_i}(v_j,1) \cup \{v_i\}$ in $H_{\sigma_H(v_i)} \setminus \{y\}$. 
Since $H \setminus \{y\}$ is Helly, there is a vertex $u$ with minimal $\sigma_H(u)$ such that $u$ suspends $M$ in $H \setminus \{y\}$, where $\sigma_H(u) > \sigma_H(y)$.
Each $v \in M$ has $\sigma_H(v) < \sigma_H(u)$.
Let $u$ be a pendant/twin to vertex $z$ in the graph $H_{\sigma_H(u)}$, where $\sigma_H(z) < \sigma_H(u)$.
Since $M \subseteq D_{H_{\sigma_H(u)}}(u,1)$ and $|M| \ge 4$, clearly, $u$ is not pendant to $z$.
Hence, $u$ is a twin to $z$ and therefore, $M \subseteq D_{H_{\sigma_H(u)}}(z,1)$.
Thus, $z$ suspends $M$, a contradiction with the minimality of $\sigma_H(u)$.
\end{proof}

\section{Graphs with exponentially large injective hulls}\label{sec:exponentialInjectiveHull}
We show that several restrictive graph classes, including split graphs, cocomparability graphs, \hnew{bipartite graphs}, and consequently graphs of bounded hyperbolicity, graphs of bounded chordality, graphs of bounded tree-length or tree-breadth, and graphs of bounded diameter can have injective hulls that are exponential in size. In particular, there is a graph $G$ of that class such that $|V(\calH(G))| \in \Omega (a^n)$ for some  constant $a>1$ and $n=|V(G)|$. 

We will use the following lemma to obtain a lower bound on the number of vertices in the injective hull of a particular graph.
Recall that a  set $S \subseteq V(G)$ is said to be a \emph{2-set} if all vertices of $S$ have pairwise distance at most 2.
\begin{lemma}\label{lem:unsuspendedTwoSets}
If $G$ has at least $k$ unsuspended maximal 2-sets, then $|V(\calH(G)) \setminus V(G)| \ge k$.
\end{lemma}
\begin{proof}
Each unsuspended maximal 2-set $S = \{v_1,\dots,v_\ell\}$ corresponds to a unique family of pairwise intersecting disks $\{D(v_i,1) : v_i \in S\}$ that have no common intersection in $G$.
As $\calH(G)$ \hnew{is the smallest Helly graph into which $G$ isometrically embeds}, then for each $S$ there is a unique Helly vertex $h \in V(\calH(G))$ universal to \hnew{maximal 2-set} $S$ in $\calH(G)$,
i.e., $h(v_i) = 1$ for each $v_i \in S$ and $h(x)=d_G(x,S)+1$ for each $x \in V(G) \setminus S$ (see Section~\ref{sec:injectiveHullProperties}).
\end{proof}

\subsection{Split graphs}
A graph is a \emph{split graph} if there is a partition of its vertices into a clique and an independent set~\cite{foldes1976split,tyshkevichSplit}. 
We construct a special split graph $G$ as follows.
Let $X=(x_1,x_2,\dots,x_k)$ be an independent set and let $Y=(y_1,y_2,\dots,y_k)$ be an independent set.
Let also $M=(u_1,v_1,w_1,z_1, u_2,v_2,w_2,z_2, \dots, u_k,v_k,w_k,z_k)$ be a clique partitioned into $k$ complete graphs $K_4$.
For each integer $i \in [1,k]$, let $x_i$ be adjacent to $u_i$ and $v_i$, and let $y_i$ be adjacent to $w_i$ and $z_i$.
Additionally, for all distinct integers $i,j \in [1,k]$, let $x_i$ be adjacent to $u_j$ and $z_j$, and let $y_i$ be adjacent to $w_j$ and $v_j$.
See Figure~\ref{fig:chordal} for an illustration.
By construction, each vertex $x_i \in X$ is within distance 2 of every vertex in the graph except $y_i$.
Every shortest $(x_i,y_i)$-path goes through $M$, but $y_i$ and $x_i$ have \hnew{no common neighbor in $M$.}
However, each $x_i$ and $y_j$ share a common vertex $v_i$.
Observe that the resulting graph $G$ has the following distance properties:
\begin{enumerate}[noitemsep, nolistsep]
  \item[-] $\forall x_i \in X, \ \forall m \in M, \ d_{G}(x_i,m) \leq 2$ via common neighbor $u_i$;
  \item[-] $\forall y_i \in Y, \ \forall m \in M, \ d_{G}(y_i,m) \leq 2$ via common neighbor $w_i$;
  \item[-] $\forall i,j \in [1,k], \ i \ne j, \ d_{G}(x_i,y_j) \leq 2$ via common neighbor $z_j$;
  \item[-] $\forall i,j \in [1,k], \ i \ne j, \ d_{G}(x_i,x_j) = 2$ via common neighbor $u_j$;
  \item[-] $\forall i,j \in [1,k], \ i \ne j, \ d_{G}(y_i,y_j) = 2$ via common neighbor $w_j$;
  \item[-] $\forall i \in [1,k], \ d_{G}(x_i,y_i) = 3$ because $y_i$ and $x_i$ have \hnew{no common neighbor in $M$}.
\end{enumerate}

\begin{figure}[h]
\begin{minipage}{\textwidth}
\begin{center}
\includegraphics[scale=0.8]{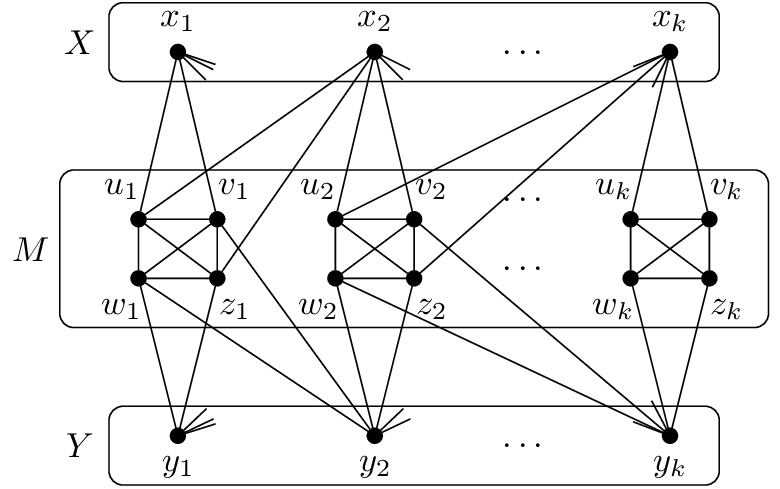}
\captionof{figure}{$G$ is a split graph that requires exponentially many new Helly vertices. For readability, some edges are not shown. $X$ and $Y$ are independent sets and $M$ is a clique of $k$ complete graphs $K_4$.}\label{fig:chordal}
\end{center}
\end{minipage}
\end{figure}

\begin{theorem} \label{thm:splitGraphSize}
There is a split graph $G$ such that $|V(\calH(G))| \ge 2^{n/6} + 2n/3 - 2$, where $n=|V(G)|$.
\end{theorem}
\begin{proof}
  Clearly, $G$ is a split graph 
  with independent set $X \cup Y$ and clique $M$.

  We first claim that $G$ described above has $2^k$ maximal 2-sets, where $k=n/6$.
  Let $S$ be a maximal 2-set in $G$.
  Since all vertices are within distance at most 2 from $M$, then $M \subset S$.
  It remains only to observe that for each $i \in [1,k]$, either $x_i \in S$ or $y_i \in S$, but not both since $d_{G}(x_i,y_i)=3$.

  We next claim that any maximal 2-set $S$ that contains at least two vertices from $X$ and at least two vertices from $Y$ is unsuspended.
  By contradiction, suppose a vertex $m \in V(G)$ suspends $S$.
  As $X$ and $Y$ are independent sets, necessarily $m \in M$.
  Thus, $m \in \{u_i, w_i, v_i, z_i\}$ for some $i \in [1,k]$.
  However, for all $j \in [1,k]$, $d_G(u_i,y_j)=2$ and $d_G(w_i,x_j)=2$ holds.
  Hence, $m \ne u_i$ and $m \ne w_i$.
  As there are at least two vertices of $X$ in $S$, there is an $x_j \in S$ such that $d_G(v_i,x_j)=2$.
  As there are at least two vertices of $Y$ in $S$, there is a $y_j \in S$ such that $d_G(z_i,y_j)=2$.
  Thus, $m \ne v_i$ and $m \ne z_i$, a contradiction with the choice of $m$.

  Moreover, there are at least $2^k - 2k - 2$ unsuspended maximal 2-sets in $G$. 
  Observe that only $2$ maximal 2-sets $S$ have no $x_i \in S$ or have no $y_i \in S$.
  There are $k$ maximal 2-sets which have only one $x_i \in S$ (one $i \in [1,k]$ is reserved for $x_i \in S$ and all other $j \in [1,k]$, $j \neq i$, have $y_i \in S$).
  Similarly, there are $k$ maximal 2-sets which have only one $y_i \in S$.
  By Lemma~\ref{lem:unsuspendedTwoSets}, $|V(\calH(G)) \setminus V(G)| \ge 2^k - 2k - 2$.
  Including the $6k$ vertices of $V(G)$, one obtains $|V(\calH(G))| \ge 2^k +4k - 2$.
\end{proof}

We remark that split graphs are chordal graphs.
Additionally, chordal graphs are 1-hyperbolic\cite{Wu2011}. 
$G$ also has tree-length $tl(G) \leq 1$\cite{DOURISBOURE20072008} and tree-breadth $tb(G) \leq 1$\cite{DBLP:journals/algorithmica/DraganK14}. 

\begin{corollary}
Split graphs, chordal graphs, $\alpha_1$-metric graphs, 
1-hyperbolic graphs, graphs with $tl(G) \le 1$, $tb(G) \le 1$, and graphs with $diam(G) \le 3$ can have exponentially large injective hulls.
Specifically, there is a graph $G$ of that class with $|V(\calH(G))| \in \Omega(a^n)$, where $a>1$ and $n=|V(G)|$.
\end{corollary}

\subsection{Cocomparability graphs}
Cocomparability graphs are exactly the graphs which admit a \emph{cocomparability ordering}~\cite{DAMASCHKE199267}, i.e., an ordering $\sigma = [v_1,v_2,\dots,v_n]$ of its vertices such that if $\sigma(x) < \sigma(y) < \sigma(z)$ and $xz \in E(G)$, then $xy \in E(G)$ or $yz \in E(G)$ must hold. Cocomparability graphs form a subclass of AT-free graphs.

A special cocomparability graph $G$ is constructed as follows.
Let $X=(x_1,x_2,\dots,x_k)$ be a clique and $Y=(y_1,y_2,\dots,y_k)$ be a clique.
Let also $M=(u_1,v_1,w_1,z_1, u_2,v_2,w_2,z_2, \dots, u_k,v_k,w_k,z_k)$ be a clique partitioned into $k$ complete graphs $K_4$.
For each integer $i \in [1,k]$, let $x_i$ be adjacent to $u_i$ and $v_i$, and let $y_i$ be adjacent to $w_i$ and $z_i$.
Additionally, for all distinct integers $i,j \in [1,k]$, let $x_i$ be adjacent to $u_j$ and $z_j$, and let $y_i$ be adjacent to $w_j$ and $v_j$.
See Figure~\ref{fig:cocomparability} for an illustration. We emphasize that the key difference between graph $G$ described above and the chordal graph construction in Figure~\ref{fig:chordal} is that, here, $X$ and $Y$ are cliques.

By construction, each vertex $x_i \in X$ is within distance 2 of every vertex in the graph except $y_i$.
Every shortest $(x_i,y_i)$-path goes through $M$, but $y_i$ and $x_i$ have \hnew{no common neighbor in $M$.}
However, each $x_i$ and $y_j$ share a common vertex $v_i$.
Observe that the resulting graph $G$ has the following distance properties:
\begin{enumerate}[noitemsep, nolistsep]
  \item[-] $\forall i,j \in [1,k], \ i \neq j, \ d_G(x_i,x_j) = d_G(y_i,y_j) = 1$;
  \item[-] $\forall x_i \in X, \ \forall m \in M, \ d_{G}(x_i,m) \leq 2$ via common neighbor $u_i$;
  \item[-] $\forall y_i \in Y, \ \forall m \in M, \ d_{G}(y_i,m) \leq 2$ via common neighbor $w_i$;
  \item[-] $\forall i,j \in [1,k], \ i \ne j, \ d_{G}(x_i,y_j) \leq 2$ via common neighbor $z_j$;
  \item[-] $\forall i \in [1,k], \ d_{G}(x_i,y_i) = 3$ because $y_i$ and $x_i$ have \hnew{no common neighbor.}
\end{enumerate}

\begin{figure}[h]
\begin{minipage}{\textwidth}
\begin{center}
\includegraphics[scale=0.8]{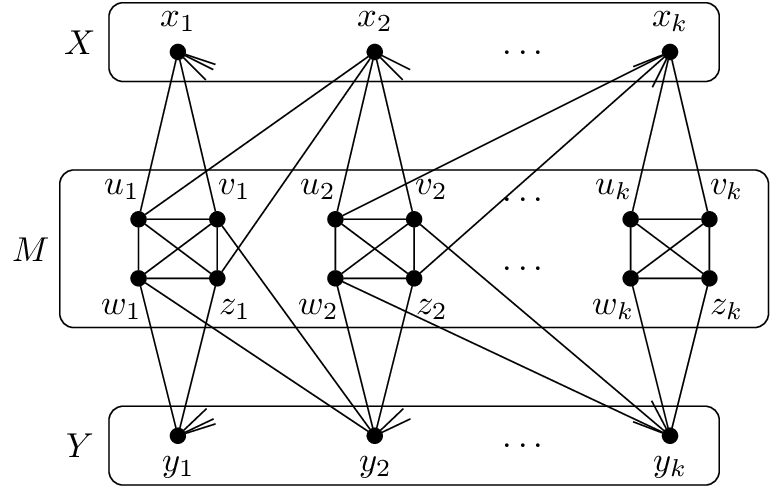}
\captionof{figure}{$G$ is a cocomparability graph that requires exponentially many new Helly vertices. For readability, some edges are not shown. $X$ and $Y$ are each cliques and $M$ is a clique of $k$ complete graphs $K_4$.}\label{fig:cocomparability}
\end{center}
\end{minipage}
\end{figure}

\begin{theorem} \label{thm:cocomparabilitySize}
There is a cocomparability graph $G$ such that $|V(\calH(G))| \ge 2^{n/6} + 2n/3 - 2$, where $n = |V(G)|$.
\end{theorem}
\begin{proof}
  The proof that $G$ has an exponential number of maximal unsuspended 2-sets is the same as in the proof of  Theorem~\ref{thm:splitGraphSize}, establishing $|V(\calH(G))| \ge 2^{n/6} + 2n/3 - 2$.
  


  It remains only to show that $G$ is a cocomparability graph.
  Let $m_1,m_2$ be two vertices of $M$.
  Let $\sigma$ be a vertex ordering of $G$ such that $\sigma(x) < \sigma(m_1)$ for all $x \in X$, $\sigma(m_1) \le \sigma(m) \le \sigma(m_2)$ for all $m \in M$, and $\sigma(m_2) < \sigma(y)$ for all $y \in Y$.
  That is, $\sigma$ is an ordering which consists of all vertices of $X$, followed by all vertices of $M$, followed by all vertices of $Y$.
  We claim that $\sigma$ is a cocomparability ordering.
  Since $X$ is a clique, for any $x_ix_j \in E(G)$ and any $x \in X$ such that $\sigma(x_i) < \sigma(x) < \sigma(x_j)$
  has an edge to $x_i$ and $x_j$.
  We apply the same argument to vertices of cliques $M$ and $Y$.
  Thus, any ordering of the vertices of $X$ alone is a cocomparability ordering, any ordering of the vertices of $M$ alone is a cocomparability ordering, and any ordering of the vertices of $Y$ alone is a cocomparability ordering.
  
  We next show that any other possible edges between the sets $X,M,Y$ satisfy the constraints of a cocomparability ordering.
  Consider any vertices $x \in X$, $m \in M$, and $v \in V(G)$ with $\sigma(x) < \sigma(v) < \sigma(m)$ and $xm \in E(G)$.
  Then, either $v \in M$ and therefore $vm \in E(G)$, or $v \in X$ and therefore $vx \in E(G)$.
  By symmetry, any $m \in M$, $v \in V(G)$, and $y \in Y$ with $\sigma(m) < \sigma(v) < \sigma(y)$ and $xm \in E(G)$
  satisfies that either $vm \in E(G)$ or $vy \in E(G)$.
  By construction, all vertices $x \in X$ and all $y \in Y$ satisfy $xy \notin E(G)$.
  Therefore, $\sigma$ is a cocomparability ordering.
\end{proof}

\begin{corollary} \label{cor:ATFreeSize}
Cocomparability graphs and AT-free graphs can have exponentially large injective hulls. Specifically, there is a graph $G$ of that class with $|V(\calH(G))| \in \Omega(a^n)$, where $a>1$ and $n=|V(G)|$.
\end{corollary}
Currently, we do not know whether there is a permutation graph $G$ with exponentially large injective hull.

\hnew{
\subsection{Bipartite graphs}
We construct a special bipartite graph $G$ with $2k$ ($k\ge 3$) vertices as follows.}
Let $X=\{x_1,x_2,\dots,x_k\}$ be an independent set and let $Y=\{y_1,y_2,\dots,y_k\}$ be an independent set. For each $i,j \in [1,k]$ and $i\neq j$, let $x_iy_j \in E(G)$. See Figure~\ref{fig:chordalBipartite} for an illustration.
Clearly, no two vertices in $X$ are adjacent and no two vertices in $Y$ are adjacent.
By construction, $G$ has the following distance properties:
\begin{enumerate}[noitemsep, nolistsep]
  \item[-] $\forall i,j \in [1,k], \ i \neq j, \ d_{G}(x_i,y_j) = 1$;
  \item[-] $\forall i,j \in [1,k], \ i \neq j, \ d_{G}(x_i,x_j) = 2$ via common neighbor $y_p$, $p \neq i,j$;
  \item[-] $\forall i,j \in [1,k], \ i \neq j, \ d_{G}(y_i,y_j) = 2$ via common neighbor $x_p$, $p \neq i,j$;
  \item[-] $\forall i \in [1,k], \ d_{G}(x_i,y_i) = 3$ as any $x_i$ is adjacent to only vertices $y_j \in Y$, $j \neq i$.
\end{enumerate}

\begin{figure}[h]
\begin{minipage}{\textwidth}
\begin{center}
\includegraphics[scale=0.8]{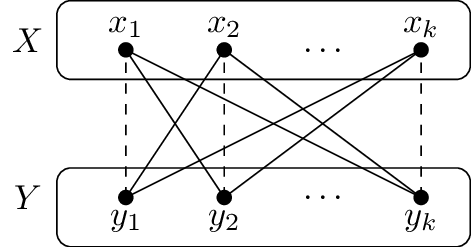}
\captionof{figure}{$G$ is a bipartite graph that requires exponentially many new Helly vertices. Non-edges are drawn in dashed lines. $X$ and $Y$ are independent sets.}\label{fig:chordalBipartite}
\end{center}
\end{minipage}
\end{figure}

\begin{theorem} \label{thm:chordalBipartiteSize}
\hnew{There is a bipartite graph $G$ with no induced $C_k$, $k >6$,} such that $|V(\calH(G))| \ge 2^{n/2} - 2$, where $n=|V(G)|$.
\end{theorem}
\begin{proof}
  \hnew{Clearly, $G$ as constructed above is bipartite and has no induced $C_k$ for $k >6$.
  }
%
  Next, we show that $G$ has exponentially many unsuspended maximal 2-sets.
  Observe that there are $2^k$ maximal 2-sets in $G$ that are suspended or unsuspended;
  for each $j \in [1,k]$, either $x_j \in S$ or $y_j \in S$, but not both since $d_{G}(x_j,y_j)=3$.
  We claim that any maximal 2-set $S$ that contains at least two vertices from $X$ and at least two vertices from $Y$ is unsuspended.
  Let $x_i,x_j,y_k,y_\ell \in S$, where $i,j,k,\ell$ are pairwise distinct integers, $x_i,x_j \in X$, and $y_k,y_\ell \in Y$.
  By construction, $x_i,x_j,y_k,y_\ell$ induce a $C_4$.
  As $X$ and $Y$ are independent sets, there is no vertex of $G$ that suspends this $C_4$ and hence $S$.
  As there are $2k + 2$ maximal 2-sets which do not contain at least two vertices from $X$ and at least two vertices from $Y$,
  by Lemma~\ref{lem:unsuspendedTwoSets}, $|V(\calH(G)) \setminus V(G)| \ge 2^k - 2k - 2$.
  Including the $2k$ vertices of $V(G)$, one obtains $|V(\calH(G))| \ge 2^k - 2$, where $k=n/2$.
\end{proof}

\begin{corollary}
Bipartite graphs can have exponentially large injective hulls. Specifically, there is a graph $G$ of that class with $|V(\calH(G))| \in \Omega(a^n)$, where $a > 1$ and $n = |V(G)|$.
\end{corollary}

\section{Conclusion}
We proved that chordal graphs, square chordal graphs, and distance-hereditary graphs are closed under Hellification;
permutation graphs are not.
We provided a \hnew{linear-}time algorithm to compute $\calH(G)$ when $G$ is distance-hereditary.
Additional graph classes are identified for which $\calH(G)$ is impossible to compute in subexponential time, including split graphs, cocomparability graphs, AT-free graphs, bipartite graphs, and graphs with a constant bound on any of the following parameters: diameter, hyperbolicity, tree-length, tree-breadth, or chordality. Recall that the \emph{chordality} of a graph $G$ is the size of its largest induced cycle; chordal graphs are exactly the graphs of chordality 3.  

A few interesting questions remain open.
As distance-hereditary graphs are square-chordal, can the injective hull of square-chordal graphs be constructed efficiently? Can the injective hull of permutation graphs be constructed efficiently?
Are cocomparability graphs or AT-free graphs closed under Hellification?

\bibliographystyle{plain}
\bibliography{bibliography}

\end{document}